\documentclass[prodmode,acmtecs]{acmsmall} 

\usepackage{graphicx}
\usepackage{amssymb,amsfonts,amsmath}
\usepackage{subfigure}
\usepackage{float}
\usepackage{xspace}
\usepackage{times}
\usepackage{hyperref}
\usepackage{color}

\usepackage{graphicx}
\usepackage{subfigure}

\usepackage{algorithm}
\usepackage[noend]{algpseudocode}
\usepackage{IEEEtrantools}
\usepackage{xspace}

\newtheorem{problem}{Problem}

\newcommand{\method}{\textit{ARRANGE}\xspace}
\newcommand{\methodExt}{\textit{d\textbf{A}ta d\textbf{R}i\-ven me\-thod fo\textbf{R} \textbf{A}s\-sessing and re\-du\-ci\textbf{NG} par\-ty frag\-m\textbf{E}n\-ta\-tion}\xspace}
\newcommand{\statusquo}{\textit{status quo}\xspace}

\newcommand{\randomsq}{\textit{random-sq}\xspace}

\newcommand{\randomdelta}{\textit{random-}$\delta$\xspace}

\newcommand{\pedro}[1]{\textcolor{black}{#1}}
\newcommand{\pedrob}[1]{\textcolor{black}{#1}}
\newcommand{\pedror}[1]{\textcolor{black}{#1}}

\begin{document}

\title{How Many Political Parties Should Brazil Have?  \\A Data-driven Method to Assess and Reduce Fragmentation \\in Multi-Party Political Systems}
\author{Pedro O.S. Vaz de Melo
\affil{Universidade Federal de Minas Gerais}
}

\begin{abstract}
In June 2013, Brazil faced the largest and most significant mass protests in a generation. These were exacerbated by the population’s disenchantment towards its highly fragmented party system, which is composed by a very large number of political parties. Under these circumstances, presidents are constrained by informal coalition governments, bringing very harmful consequences to the country. In this work I propose \method, a \methodExt in a country. \method uses as input the roll call data for congress votes on bills and amendments as a proxy for political preferences and ideology. With that, \method finds the minimum number of parties required to house all congressmen without decreasing party discipline. When applied to Brazil’s historical roll call data, \method was able to generate \pedror{$23$} distinct configurations that, compared with the \statusquo, have (i) a significant smaller number of parties, (ii) a higher discipline of partisans towards their parties and (iii) a more even distribution of partisans into parties. \method is fast and parsimonious, relying on a single, intuitive parameter.
\end{abstract}


\maketitle

\section{Introduction}

In June 2013, Brazil faced the largest and most significant mass protests in a generation, comparable in size to the protests that triggered the collapse of the military dictatorship in 1984 ~\cite{Saad-Filho2013}. The 2013 protests had been exacerbated  by the broader disenchantment of the population towards the party system in Brazil~\cite{Mische2013}. Banners with sentences such as ``no party represents me'' or ``we don't have a party,  we are Brazil!'' were commonly seen among the protesters. In response to these protests, the government proposed a program of political reform~\cite{Saad-Filho2013}. However, more than one year has passed and very little has been done.

Why are people so unhappy with the Brazilian party system? To illustrate its incapability, consider the following examples. 19 months before Rio de Janeiro stages South America's first Olympic games, an Evangelical pastor without any link to sports was nominated as Brazil's new sports minister. He replaced communist Aldo Rebelo, who oversaw preparations for the World Cup and was highly criticized for subsidizing the construction of white-elephant football stadiums~\cite{DOWNIE2014}. Aldo Rebelo, who in 1994 proposed a bill that prohibits the adoption of any technological innovation in local, state and federal agencies, is, ironically, Brazil's new minister of science and technology. Moreover, in 2010 elections, Tiririca, a well known entertainer whose career began as a circus clown in Brazil, was first elected to represent S\~{a}o Paulo in Congress, winning the most votes of any candidate in the country with the slogans ``It can't get any worse'' and ``What does a Congressman do? In fact, I do not know, but vote for me and I will tell you''.

One of the main causes of  Brazil's political inefficiency is its highly fragmented party system~\cite{mainwaring1999rethinking}. This is a system with many political parties and with no one party being able to obtain an absolute majority in the representative assembly. The more fragmented the party system is, the less likely it is that the president's party will control a majority of seats in the legislature. Therefore, presidents are usually forming informal coalition governments, needing to build cross-party coalitions to implement most major policies~\cite{mainwaring1997presidentialism}. Under these circumstances, many (if not most) deputies spend the bulk of their time arranging jobs and pork-barrel projects for their constituents in exchange for legislative support~\cite{ames2009deadlock}. Also, parties rarely organize around national-level questions, which means that Congress rarely deals with serious social and economic issues~\cite{ames2009deadlock}. As a consequence, individualism, clientelism, and personalism, rather than programmatic appeals, dominate electoral campaigns~\cite{ames2009deadlock}. More generally, party system fragmentation impacts the electoral dynamic, the process of coalition formation, governing, and ultimately, the survival of political systems in presidential democracies~\cite{Bohn2014}. In Brazil, party system fragmentation has reached one of the highest levels ever found in the world~\cite{Limongi2009,Anckar2000}. After 2014 elections, the number of parties represented in Congress grew from 23 to 28.

Besides Brazil, many countries have party systems with high levels of fragmentation, such as  Bolivia, Bulgaria, Denmark, Ecuador, Finland, France, Guatemala, India, Israel, Italy, Netherlands and Thailand~\cite{Anckar2000}. In theory, the number of parties of a country can be explained by its electoral and social structures~\cite{Colomer2005,Neto1997,TAAGEPERA2006}. With regard to the electoral structure, while plurality elections favor two-party competition, proportional representation (PR) electoral systems create fragmented party systems~\cite{duverger1963political,mainwaring1999rethinking}. Concerning the social structures,  the more socially heterogeneous a country is, the more electoral parties it will have~\cite{lipset1967party,Stoll2007}. Social heterogeneity is measured either by the number of linearly independent ideological dimensions (e.g. religious and socio-economic) being discussed in the society or the number of social cleavages (e.g. centre-periphery and state-church) a country has~\cite{Colomer2005,Stoll2007}. Thus, is it possible to measure if a fragmented party system is a reflex of a socially heterogeneous society? In this direction, how can we determine whether an electoral system is optimally fragmented? And if it is not optimally fragmented, how can we optimize it?

To answer these questions, I propose \method, a \methodExt in a country. \pedro{Inspired by the broad spectrum and advances of data analysis methods~\cite{Wu2014a,Silva2014c,Silva2014b,ICWSM1510570,Zhang2011,Yin2014,DeVries2014,Shor2011,poole2011ideology,Clinton2004,Sobkowicz2012,Pennacchiotti2011,ICWSM148073,Gao2015a,Newman2011}}, \method uses as input the roll call data, i.e., the votes given by congressmen on bills and amendments, as a proxy for political preferences and ideology. The idea, regularly employed by political scientists~\cite{poole2011ideology,Shor2011}, is that congressmen who give the same vote regularly share the same political ideology and, therefore, should belong to the same party. Using this insight, \method reorganizes the party system of a given country by trying to find the minimum amount of social cleavages that divides its congressmen into coherent voting blocks, i.e., sets of congressmen whose members votes similarly. These coherent voting blocks would be the new political parties of the analyzed country and would allow us to assess its actual level of fragmentation. If \method divides the congressmen into a much lower number of political parties than the actual number, then it is possible to conclude that party fragmentation in this particular country is much higher than it should be. If this is the case, \method immediately provides a new congressmen--party configuration that both reduces the country's level of party fragmentation and increases intra-party similarity, which could potentially increase  the efficiency of the party system~\cite{Limongi2009,mainwaring1999rethinking,mainwaring1997presidentialism,ames2009deadlock}. \pedro{To the best of my knowledge, this is the first work that proposes a data-driven method to assess and potentially reduce the number of parties in political systems.} In summary, the main contributions of this paper are:

\pedro{
\begin{enumerate}
	\item \textbf{A new problem is addressed}: what is the minimum number of political parties a country should
have given the roll votes of its congressmen? Again, although party fragmentation has been extensively studied in
the literature, to the best of my knowledge, this is the first time roll call data has been used to assess it.
\item \textbf{\method}, a fast and parsimonious method that receives as input roll call data and
outputs a new party system configuration that potentially reduces its actual level of fragmentation.
\item \textbf{It is shown that Brazil has and had many ideologically redundant parties}, i.e., parties that are similar
in the ideological space. Thus, if today Brazil has one of the highest levels of party system fragmentation in
the world (more than 20 parties), this work proves it can be much lower (down to \pedror{5} parties).
\end{enumerate}
}

\section{Fundamentals and Related Work}

\pedro{The constant advancement of information systems allows, at a growing rate, more data to be stored and generated from the most diverse situations. It is fascinating that, behind all these data, we see the reflection of the environment itself. In order to find knowledge in this invaluable evolving database, a growing number of data-driven methods are being proposed along various research areas. For instance, there are data-driven methods to predict hospital mortality from instance-based patient data~\cite{Wu2014a} and flu epidemics~\cite{Lazer2014}. In the social sciences, Silva et. al. proposed two data-driven methods to quantitatively characterize cultural behaviors of geographical regions\cite{Silva2014c,Silva2014b} and Park et. al.~\cite{ICWSM1510570} designed and evaluated a measure that captures diversity of musical tastes from social media data. In economics, \cite{DeVries2014} proposed a data-driven approach to understand online consumer behavior and engagement with brands. For the benefit of the industry sector, there are data-driven methods to monitor industrial processes~\cite{Yin2014} and to assist the development and deployment of intelligent transportation systems~\cite{Zhang2011}.}
\pedro{The constant advancement of information systems allows, at a growing rate, more data to be stored and generated from the most diverse situations. It is fascinating that, behind all these data, we see the reflection of the environment itself. In order to find knowledge in this invaluable evolving database, a growing number of data-driven methods are being proposed along various research areas. For instance, there are data-driven methods to predict hospital mortality from instance-based patient data~\cite{Wu2014a} and flu epidemics~\cite{Lazer2014}. In the social sciences, Silva et. al. proposed two data-driven methods to quantitatively characterize cultural behaviors of geographical regions\cite{Silva2014c,Silva2014b} and Park et. al.~\cite{ICWSM1510570} designed and evaluated a measure that captures diversity of musical tastes from social media data. In economics, \cite{DeVries2014} proposed a data-driven approach to understand online consumer behavior and engagement with brands. For the benefit of the industry sector, there are data-driven methods to monitor industrial processes~\cite{Yin2014} and to assist the development and deployment of intelligent transportation systems~\cite{Zhang2011}.}

\pedro{In the political sciences, data analysis methods} from roll votes primarily focuses on the estimation of cleavages and ideologies across congressmen and parties~\cite{Poole1991,Shor2011,poole2011ideology} to characterize and predict legislative behavior~\cite{Snyder2001,Clinton2003,Clinton2004}. Recently, with the advancement of political weblogs and online social networks, researchers are also extracting political knowledge from user generated data on the Web. There are studies that focus on mining political opinions~\cite{Sobkowicz2012} and profiles~\cite{Pennacchiotti2011} from the texts users post on social media applications. Others, such as~\cite{Fang2012}, extract political opinions from general texts, such as statement records of U.S. senators and online news. More recently, Leman Akoglu~\cite{ICWSM148073} classified the political polarity of individuals using roll call votes of U.S. congressmen and texts posted on political forums. The idea behind all these studies is that political preferences tend to be stable over time and can be predicted accurately.

\pedro{Still in the political sciences, the study of party systems is one of its largest sub-fields~\cite{Neto1997}}. Within this sub-field, since Duverger's seminal paper~\cite{duverger1963political}, many studies focused on predicting and understanding the factors that determine the number of parties that compete in a given polity~\cite{Riker76thenumber,TAAGEPERA2006,Samuels2008,Bohn2014,Anckar2000}. In summary, there are two lines of thought: one that emphasizes the role of electoral laws in structuring coalition incentives, and another that emphasizes the importance of preexisting social cleavages. Another fundamental problem in this sub-field is to count the number of parties by taking into account their relative size~\cite{LaaksoTaagepera79}. If, for instance, a party has a very small percentage of seats in Congress (e.g. one seat of one thousand), then it should be counted accordingly. The metric that considers this is called \textit{the effective number of parties}. Conceptually, the effective number of parties is simply the number of ``viable'' or ``important'' political parties in a party system that includes parties of unequal sizes. Since Laakso and Taagepera's seminal work~\cite{LaaksoTaagepera79}, several ways of computing the effective number of parties were proposed~\cite{Molinar1991,Dunleavy2003,Golosov2009}. The number of effective parties is a frequent metric for assessing party system fragmentation in a country~\cite{Anckar2000}.

The high interest in these problems comes from the fact that the actual number of parties usually determines the number of \textit{effective} parties, or how \textit{fragmented} a party system is~\cite{LaaksoTaagepera79}. Highly fragmented party systems can affect governance drastically~\cite{Raile2010}. The more fragmented the party system is, the less likely it is that the president's party will control a majority of seats in the legislature. Simone Bohn~\cite{Bohn2014} reviewed the literature and concluded that party system fragmentation impacts the electoral dynamic, the process of coalition formation, governing, and ultimately, the survival of political systems in presidential democracies. Thus, in this paper, we measure party fragmentation by counting both the actual and the effective number of parties.

Another crucial factor for governance is party discipline, i.e., the ability of a political party to get its members to support the policies of their party leadership. Mainwaring and Shugart~\cite{mainwaring1997presidentialism} assessed the effects of this on the costs of governing. If parties are not disciplined, presidents will be forced to rely on ad-hoc coalitions based on the distribution of patronage to individual legislators, which raises the costs of governing and reduces policy coherence. Limongi and Figueiredo~\cite{Limongi2009} argued that ``institutional engineering'' should focus on electoral formulas that reduce party fragmentation and increase party discipline. Brazil is a special case in politics for its high level of party fragmentation, being consistently analyzed in the literature~\cite{Desposato2006,Samuels2008,Limongi2009,ames2009deadlock,Raile2010,Bohn2014}. Thus, using Brazil as the use case makes this work specially challenging but rewarding.

\pedro{Finally, it is important to emphasize that this work differs significantly from those that focus on algorithms to find communities in networks~\cite{Newman2011,Fortunato2010}. Although such algorithms could be applied here to detect communities of congressmen that are ideologically similar and, therefore, could compose a political party, this is very different from our problem in two major aspects. First, here our goal is to find the minimum number of communities (in our case, political parties), which is an optimization problem not addressed by community detection algorithms, which usually aim to maximize modularity~\cite{Newman2004} or any other cohesion metric~\cite{Aldecoa2011}. Second, while traditional community detection algorithms do not allow two disjoint subgraphs to be part of the same community, our major constraint here is party discipline. Thus, here we allow two disjoint subgraphs (in our case, two ideologically dissimilar groups of congressmen) to be part of the same community if the party discipline constraint is satisfied. The comparison between the communities (political parties) generated by the proposed method and by the state of the art community detection algorithms is left for future work.}

\section{Data Description}
\label{sec:data}

All the data used in this work was collected from the Open Data (\textit{Dados Abertos}) project of the House of Representatives (or Chamber of Deputies) of Brazil. In total, I collected \pedror{$744,195$ thousand} roll votes on \pedror{$774$} bills of $1,582$ thousand congressman that worked in the House of Representatives of Brazil from November, 4th, 1998 to December, 3rd, 2014. The reason for this particular time interval is related to the purpose of this work. Since party discipline is a fundamental metric of evaluation, I only collected the bills in which the party leaders declared the desired vote for their fellow partisans. More than \pedror{$95\%$} of the votes given during this period had a declared party leader vote. Moreover, note that congressmen vote for bills and their amendments. An amendment is a proposition presented as ancillary to the bill, to amend its form or content. Thus, bills and amendments compose a total of \pedror{$2,162$} thousand propositions to be voted by the congressmen. Each congressman may or may not agree with the vote of his/her party leader. There are, in total, \pedror{$35,216$} thousand declared votes of the leaders of the 36 parties that had congressmen elected for the House of Representatives during the analyzed period. 
 
\section{Formal Definitions}

As discussed previously, reducing party fragmentation only makes sense if party discipline does not decrease significantly. One way to measure the discipline of a congressman is to compute the fraction of votes given by him/her through his political life that agreed with his/her party leader. However, in Brazil it is common for congressmen to switch parties. Also, it is well known that some parties demand (or inspire) higher levels of discipline than others~\cite{Limongi2009}. Thus, instead of analyzing the discipline of the whole political life of a congressman, I will analyze his/her discipline as a member of a single party, i.e., a congressman may have different levels of discipline if he/she was member of different parties during his/her political life. For simplicity, from now on I will assume that the vote given by the leader of the party was given by the party itself, e.g., I will call \textit{party vote} the vote given by the leader of the party.

Thus, given the set of $m$ congressmen $\mathcal{U} = \{u_1, u_2, ..., u_m\}$ and the $N$ political parties that compose the set $\mathcal{P} = \{p_1, p_2, ... p_N\}$, I define a \textbf{partisan} $a:=(u,p)$ as the tuple formed by a congressman $u$ and a political party $p$. The set containing all $M$ partisans is defined as $\mathcal{A} = \{a_1, a_2, ... a_M\}$. These partisans' job is to vote for the set $\mathcal{B} = \{b_1, b_2, ..., b_n\}$  of $n$ bills and amendments that were put to vote in the House of Representatives during the analyzed period. From now on, I will use the term \textbf{propositions} to refer to both bills and amendments. Since a partisan $a$ had not necessarily voted for all propositions, I define the set $\mathcal{B}_a \subseteq \mathcal{B}$ as the set of propositions that were voted by partisan $a$. I also define $\mathcal{A}_p$ as the set of partisans which are members of party $p$, i.e., $\mathcal{A}_p = \{a_i | a_i := (u,p') \wedge p'=p\}$.

Before the partisans give their votes for a given proposition $b$, their parties have to announce their votes for $b$. Again, since a party $p$ had not necessarily voted for all propositions, I define the set $\mathcal{B}_p \subseteq \mathcal{B}$ as the set of propositions that were voted by party $p$. For each proposition $b \in \mathcal{B}_p$ of a given party $p$, there is a vote $v_p^b$ associated with it. In the same way, for each proposition $b' \in \mathcal{B}_a$ of a given partisan $a$, there is a vote $v_a^{b'}$ associated with it. Thus, a party $p$ and a partisan $a$ have, respectively, a set of votes $\mathcal{V}_p$ and $\mathcal{V}_a$, where $|\mathcal{V}_p| = |\mathcal{B}_p|$ and $|\mathcal{V}_a| = |\mathcal{B}_a|$. The set of all votes given by partisans and parties is simply $\mathcal{V}$.

A vote $v_p^b$ given by party $p$ on a proposition $b$ may be of four types: \textit{Y} (yes), \textit{N} (no), \textit{O} (obstruction) and \textit{F} (free), i.e., $v_p^b \in \{Y, N, O, F\}$. If the vote is $Y$ ($N$), the party approves (disapproves) the proposition. If the vote is $O$, the party is trying to avoid the vote on the proposition, i.e., its partisans are called to withdraw from the plenary. Finally, if the vote is $F$, its partisans are free to vote at will. Similarly, a partisan $a$ vote $v_a^b$ on a proposition $b$ may be of three types: \textit{Y}, \textit{N} and \textit{O}, i.e., $v_a^b \in \{Y, N, O\}$. If the vote is $Y$ ($N$), the partisan approves (disapproves) the proposition. If the vote is $O$, the partisan withdrew from plenary.

In this work, the vote is the fundamental feature that determines a preference or ideology. Since for all propositions in our dataset we have both the vote of the partisan and the party, here I define a general function $agrees(v_1, v_2)$ that receives two votes as input and outputs $1$ if the votes are in accordance or $0$ otherwise. This function is defined as: 

\begin{equation}
agrees(v_1, v_2) = \begin{cases}
1 & \text{ if } v_1 = v_2 \text{ or } v_1=F \text{ or } v_2=F \\ 
0 & \text{ otherwise. } 
\end{cases} \label{eq:voteagrees}
\end{equation}

Note that both $v1$ and $v2$ can be a vote of a partisan or a party. Also note that a $F$ vote implies accordance, since the party that gave that vote does not particularly care about its members’ votes. Once we know how to compare votes, we can propose a way to compare the similarity $sim(i,j)$ between two vote sets $V_i$ and $V_j$. It is given as: 
\begin{equation}
	sim(i,j) = {\sum_{b~\in~B_i \cap B_j } agrees(v_i^b, v_j^b)} \times {\left | B_i \cap B_j \right |}^{-1}. \label{eq:sim}
\end{equation}

In summary, $sim(i,j)$ sums all the votes in agreement between the vote sets $V_i$ and $V_j$ considering only the propositions that are in both sets. From Equation~\ref{eq:sim}, we can define the three levels of discipline that we will use throughout this paper. First, I define \textbf{partisan discipline} as the discipline $d_{a\rightarrow p}$ of a partisan $a:=(u,p)$ towards his/her party $p$, calculated as:
\begin{equation}
d_{a\rightarrow p} = sim(a,p) ~|~ a:=(u,p). \label{eq:partisandiscipline}
\end{equation}
Second, I define \textbf{party discipline} as the discipline $d_p$ throughout all the votes that were given by partisan members of the party $p$, calculated as: 
\begin{equation}
d_p = \frac{\sum_{a \in \mathcal{A}_p}{|\mathcal{B}_a|}\times d_{a\rightarrow p}}{\sum_{a \in \mathcal{A}_p}|\mathcal{B}_a|} \label{eq:partydiscipline}
\end{equation}
Finally, I define \textbf{overall discipline} as the discipline $d_*$ throughout all the votes given in the House of Representatives during the analyzed period, calculated as:
\begin{equation}
d_* = \frac{\sum_{[a:=(u,p)] \in \mathcal{A}}{|\mathcal{B}_a|}\times d_{a\rightarrow p}}{\sum_{[a:=(u,p)] \in \mathcal{A}}|\mathcal{B}_a|} \label{eq:overaldiscipline}
\end{equation}

\section{Politics in Brazil}

In this section I will show a summarized view of the fundamental characteristics of Brazilian party system that are relevant to the purpose of this work. \pedro{Because all parties in Brazil are often referred by their acronyms, I will also use their acronyms instead of their names. Thus, please refer to Table~\ref{tab:parties} for a list of all parties' names and their respective acronyms and sizes, in this case given by the size of their $\mathcal{B}_p$}s.  

\begin{table}
\tbl{Current and historical parties of Brazil.\label{tab:parties}}{
    \begin{tabular}{|c|c|c|}
\hline		
\textbf{Name} & \textbf{Acronym} & \textbf{$|\mathcal{B}_p|$}\\ \hline 
Partido do Movimento Democr\'{a}tico Brasileiro & PMDB & \pedror{2,163} \\ \hline 
Partido da Social Democracia Brasileira & PSDB  & \pedror{2,163}  \\ \hline 
Partido dos Trabalhadores & PT  & \pedror{2,163} \\ \hline 
Partido Democr\'{a}tico Trabalhista & PDT  & \pedror{2,152} \\ \hline 
Partido Socialista Brasileiro & PSB & \pedror{2,144}  \\ \hline 
Partido Trabalhista Brasileiro & PTB  & \pedror{2,144} \\ \hline 
Partido Popular Socialista & PPS  & \pedror{2,137}  \\ \hline 
Partido Comunista do Brasil & PCDOB  & \pedror{2,127} \\ \hline 
Partido Verde & PV & \pedror{1,798} \\ \hline 
Partido Progressista & PP  & \pedror{1,590} \\ \hline 
Partido Social Cristão & PSC  & \pedror{1,353}  \\ \hline 
Partido Socialismo e Liberdade & PSOL & \pedror{1,230} \\ \hline 
Partido da República & PR & \pedror{1,134} \\ \hline 
Partido da Frente Liberal & PFL & \pedror{1,081}   \\ \hline 
Democratas & DEM & \pedror{1,077} \\ \hline 
Partido Republicano Brasileiro & PRB  & \pedror{1,061} \\ \hline 
Partido Liberal & PL  & \pedror{1,001}   \\ \hline 
Partido da Mobilização Nacional & PMN & \pedror{981}  \\ \hline 
Partido Social Liberal & PSL & \pedror{865} \\ \hline 
Partido Humanista da Solidariedade & PHS & \pedror{662} \\ \hline 
Partido Trabalhista Cristão & PTC  & \pedror{606} \\ \hline 
Partido Progressista Brasileiro & PPB & \pedror{567} \\ \hline 
Partido Trabalhista do Brasil & PTDOB  & \pedror{471} \\ \hline 
Partido Social Democr\'{a}tico & PSD  & \pedror{466}  \\ \hline 
Partido Social Trabalhista & PST  & \pedror{464} \\ \hline 
Partido de Reedificação da Ordem Nacional & PRONA & \pedror{442} \\ \hline 
Partido Republicano Progressista & PRP  & \pedror{356} \\ \hline 
Partido Renovador Trabalhista Brasileiro & PRTB & \pedror{179} \\ \hline 
Partido Republicano da Ordem Social & PROS & \pedror{147} \\ \hline 
Solidariedade & SDD & \pedror{146} \\ \hline 
Partido Ecol\'{o}gico Nacional & PEN & \pedror{100} \\ \hline 
Partido Trabalhista Nacional & PTN & \pedror{92} \\ \hline 
Partido dos Aposentados da Na\c{c}\~{a}o & PAN & \pedror{75} \\ \hline 
Partido Social Democrata Crist\~{a}o & PSDC & \pedror{38}  \\ \hline 
Partido Republicano Brasileiro & PMR  & \pedror{33} \\ \hline 
Partido Socialista dos Trabalhadores Unificados & PSTU & \pedror{8}\\ \hline 
    \end{tabular}}
\end{table}

\pedro{First, in} Fig.~\ref{fig:partysize}, I show the historical participation of all parties in the House of Representatives during the analyzed period. \textit{Party participation}, which I will interchangeably call \textit{party size}, is represented by the total number of propositions that were voted by the members of the party (horizontal axis) and the total number of partisans that are and were members of the party (vertical axis). Observe the heterogeneity of this universe. From the three biggest parties (PMDB, PSDB and PT), with hundreds of partisans who voted for thousands of propositions, to the two smallest ones (PSTU and PMR), which together have only three partisans and \pedror{forty one} propositions, there are another $31$ parties with very distinct levels of representation.  

\begin{figure}[!hbt]
\centering
  {\includegraphics[width=.65\textwidth]{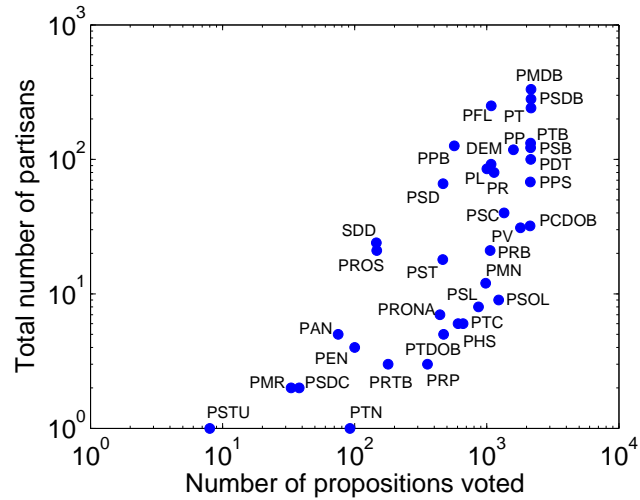}}
  \caption{Historical parties' size in Brazil.}
  \label{fig:partysize}
\end{figure}

In order to verify if there is any correlation between participation and party discipline, I show in Fig.~\ref{fig:partydiscipline}, for each party $p$, the total number of votes given by partisans members of $p$ (horizontal axis) and the party discipline $d_p$ (vertical axis). Note that there is no apparent relationship between party discipline and participation. In fact, the Pearson's correlation coefficient between party discipline and the total number of votes is \pedror{$0.24$}, but since the $p-value$ for testing the hypothesis of no correlation is \pedror{$0.28$}, we cannot affirm the correlation is significant. Nevertheless, as already observed by~\cite{Limongi2009} using a smaller dataset, party discipline in Brazil is consistently high: no party has a historical party discipline below $0.75$ and only \pedror{two} parties have a figure below $0.85$. 

\begin{figure}[!hbt]
\centering
  {\includegraphics[width=.65\textwidth]{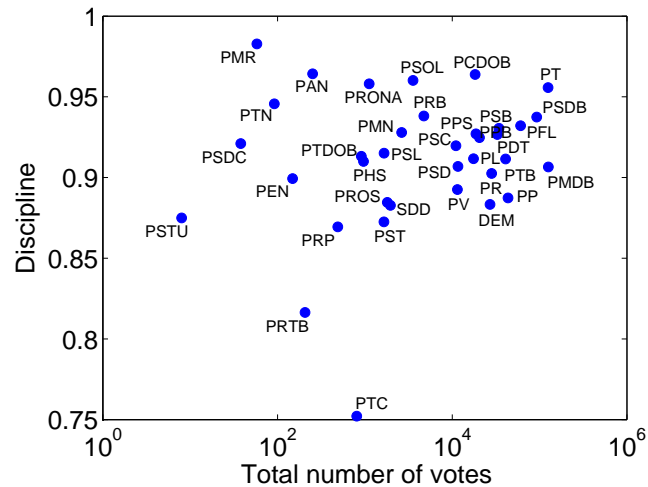}}
  \caption{Historical parties' discipline in Brazil.}
  \label{fig:partydiscipline}
\end{figure}

In Fig.~\ref{fig:partisandiscipline}, we show the behavior of partisans discipline in Brazil during the analyzed period. First, observe in Fig.~\ref{fig:partisandiscipline}a the Cumulative Distribution Function (CDF) of all partisans' disciplines. Note that the curve representing partisan discipline in Brazil is not very far away from the ideal curve, where all partisans have discipline of $1.0$. Thus, together with party discipline, partisan discipline is also usually high in Brazil: only \pedror{$6.3\%$} of partisans have discipline lower or equal to $0.8$. In Fig.~\ref{fig:partisandiscipline}b, we plot the heatmap of partisans using their disciplines and total number of votes. The color bar at the right indicates the number of partisans in a given area of the map. Observe that the vast majority of partisans are located in the upper-right of the heatmap, i.e., they have given many votes and have high partisan discipline. \pedror{The Pearson's correlation coefficient of $0.03$ and $p-value=0.23$ indicate that there is no correlation between discipline and participation.}

\begin{figure}[!hbt]
\centering
  \subfigure[CDF]
            {\includegraphics[width=.45\textwidth]{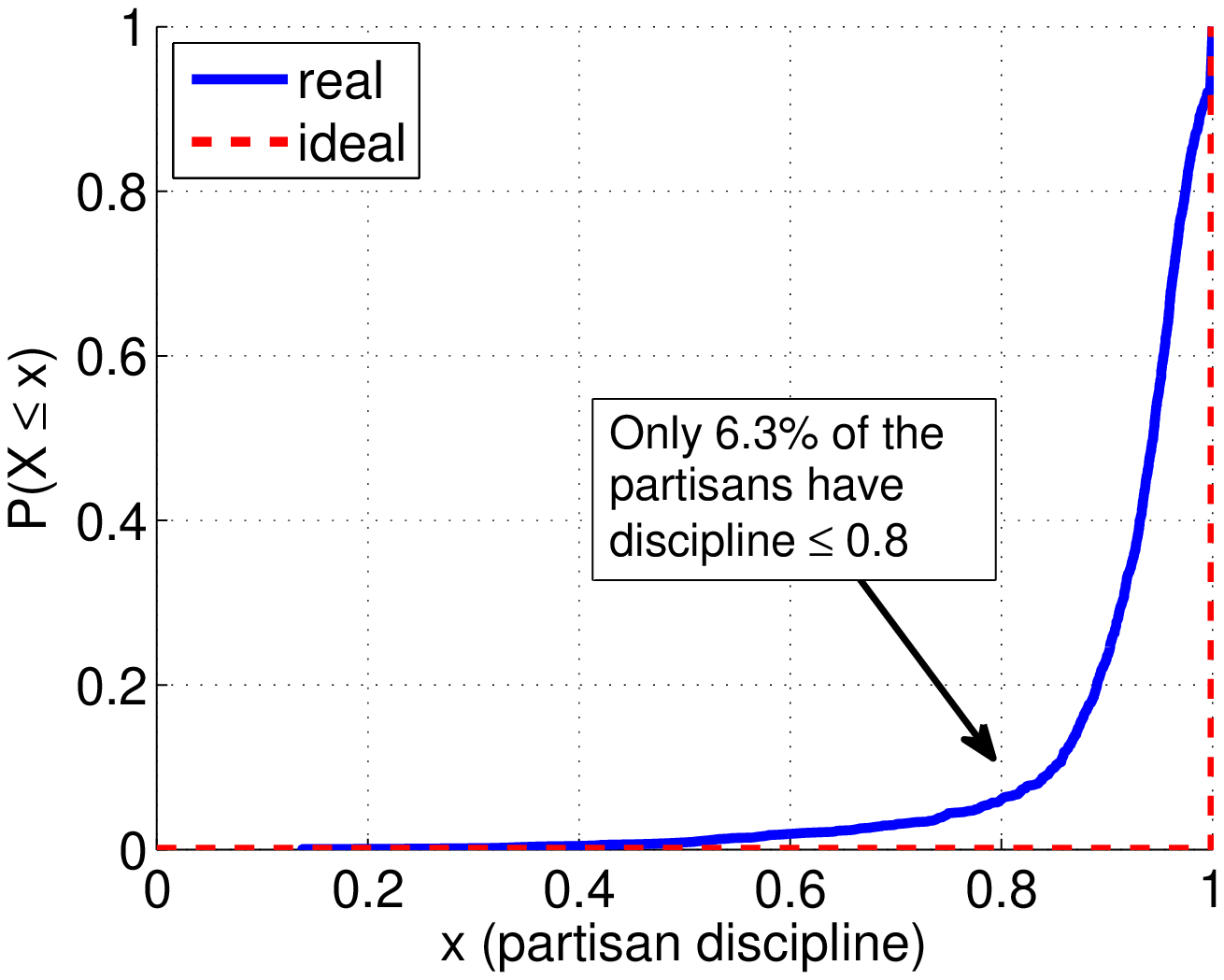}}
	\subfigure[Heatmap]
            {\includegraphics[width=.45\textwidth]{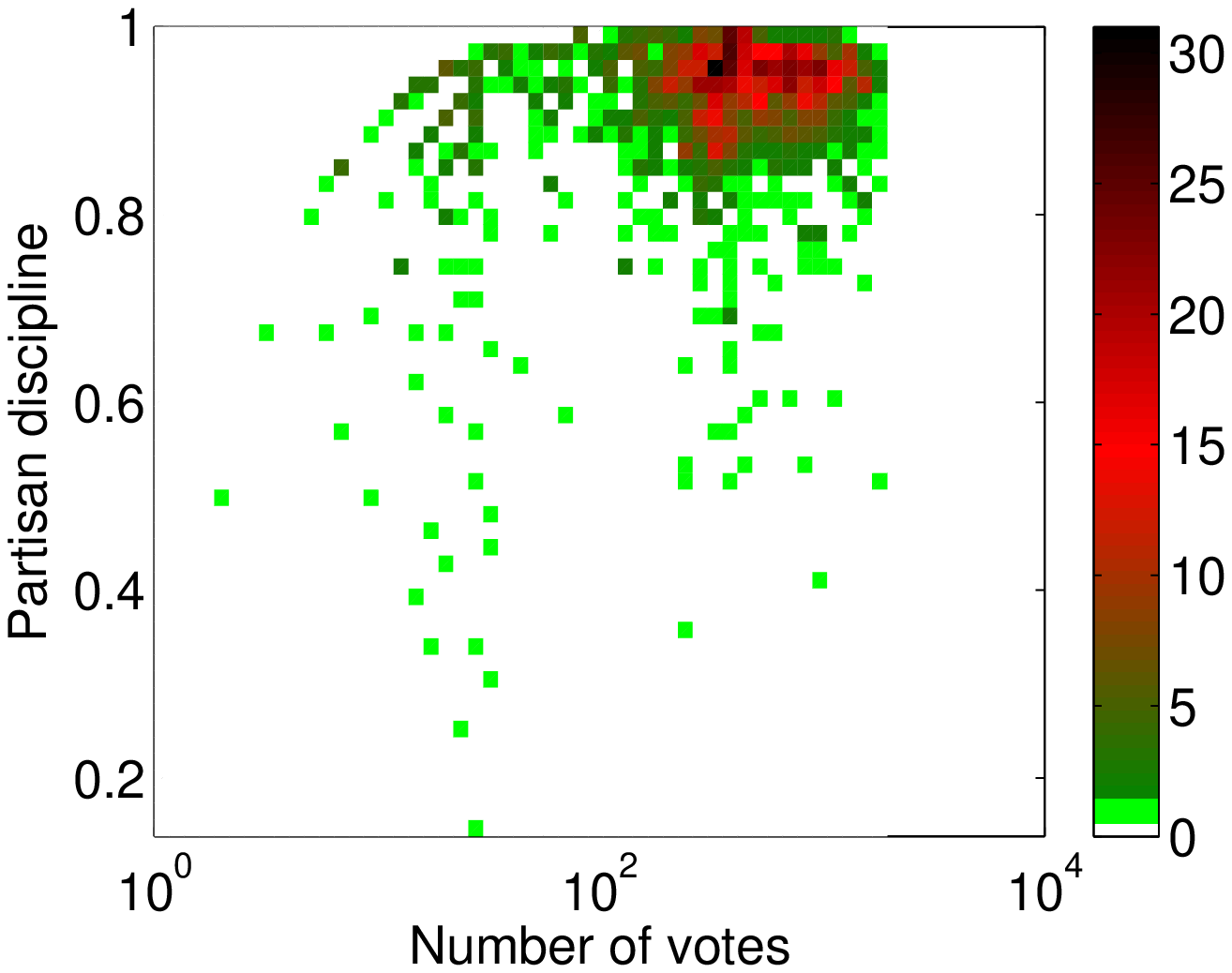}}
  \caption{Partisan's discipline in Brazil.}
  \label{fig:partisandiscipline}
\end{figure}

Although party fragmentation in Brazil has reached one of the highest levels ever found in the world~\cite{Limongi2009,Anckar2000}, we have seen that party and partisan discipline is consistently high. Does this mean that the level of party fragmentation in Brazil is necessary? According to the seminal \textit{doctrine of responsible party government}~\cite{ranney1975curing}, parties must differ sufficiently between themselves, providing the electorate with a proper range of choice between alternative actions. Given that, I reformulate the previous question: is the actual level of party fragmentation in Brazil a consequence of a high number of sufficiently different parties?

Instead of performing a deep clustering analysis to answer this question, I will apply the Principal Component Analysis (PCA)~\cite{pca} technique to the matrix $M_{\mathcal{V}_{BR}}$ composed by the voting vectors $\mathcal{V}_{a_i}$ of each partisan $a_i$. PCA is a widely used statistical technique for unsupervised dimension reduction. It transforms the data into a new coordinate system such that the greatest variance is achieved by projecting the data into the first coordinate, namely principal component, the second greatest variance is achieved by projecting into the second coordinate, the second component, and so on. In order to draw more interesting conclusions from this analysis, I will also apply PCA to the matrix $M_{\mathcal{V}_{US}}$ composed by the voting vectors of the U.S. congressmen. For this, I will use the same dataset used in~\cite{ICWSM148073}, which consists of votes on 1655 bills in The House of Representatives in years 2009-2010 by 451 representatives. In the matrices the $YES$ votes are represented as $1$, the $NO$ votes as $-1$, and the $F$, $O$ and non-attendance as $0$.
To make the comparison more precise, I will only use the votes in years 1999-2000 for constructing $M_{\mathcal{V}_{BR}}$. This period comprises votes on the \pedror{349} bills by 767 contemporary partisans of 18 parties in the first two years of president Cardoso in power. Also, it is the two year period that gives the higher explained variance by the first two components of the PCA. 

In Fig.~\ref{fig:PCA}, I show the first two components of the PCA for both $M_{\mathcal{V}_{BR}}$ and $M_{\mathcal{V}_{US}}$, where each point represents a congressman and each symbol represents a party. Observe that, for both USA and Brazil the first two components explain a significant part of the variance: $67\%$ and \pedror{$53\%$}, respectively. However, only the USA PCA can visually divide the members of different parties. For Brazil, many members of different (same) parties are located together (apart). This suggests that parties in Brazil are not sufficiently different to justify one of the highest levels of party fragmentation ever found in the world~\cite{Limongi2009,Anckar2000}.

\begin{figure}[!hbt]
\centering
  \subfigure[USA]
            {\includegraphics[width=.45\textwidth]{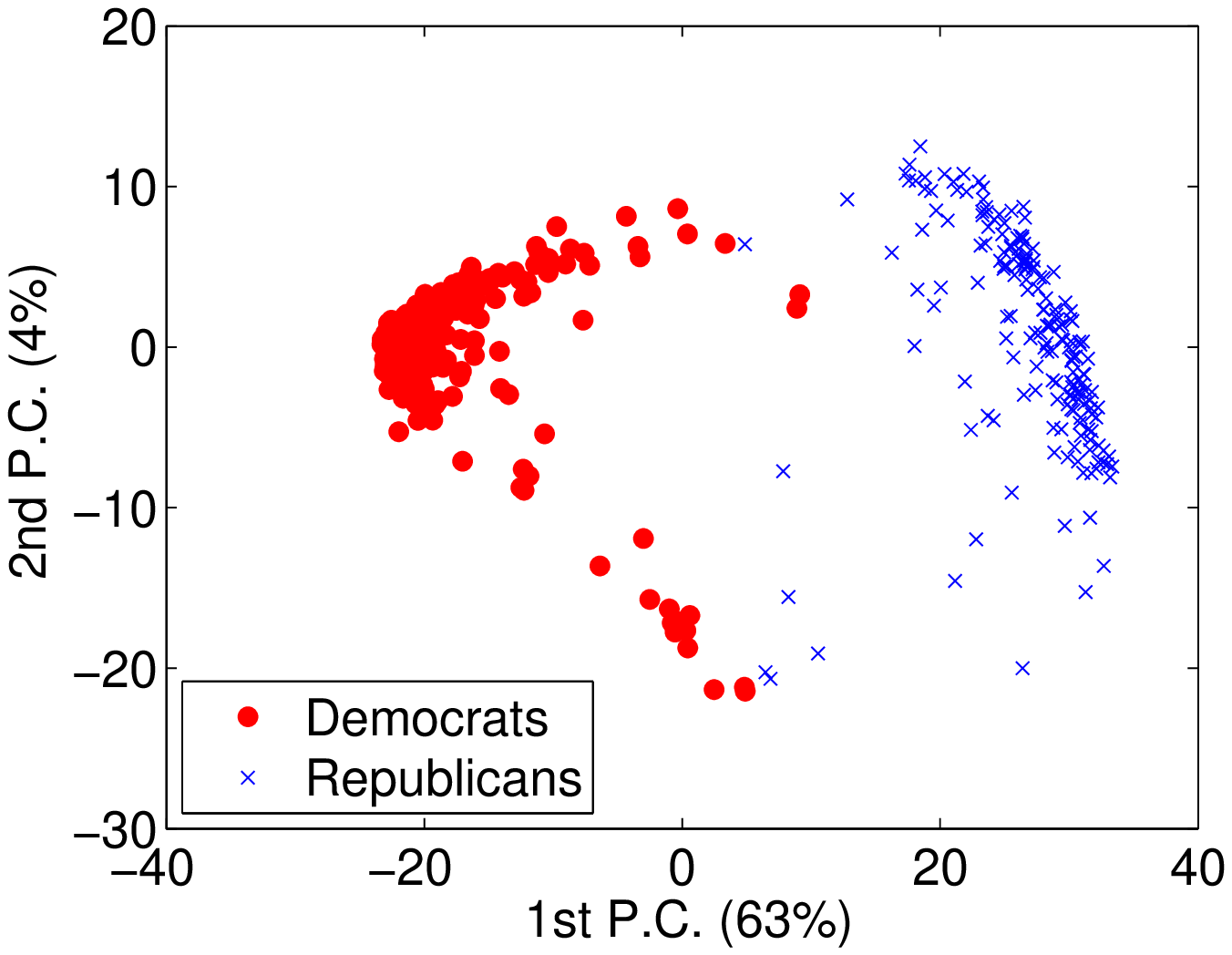}}
	\subfigure[Brazil]
            {\includegraphics[width=.45\textwidth]{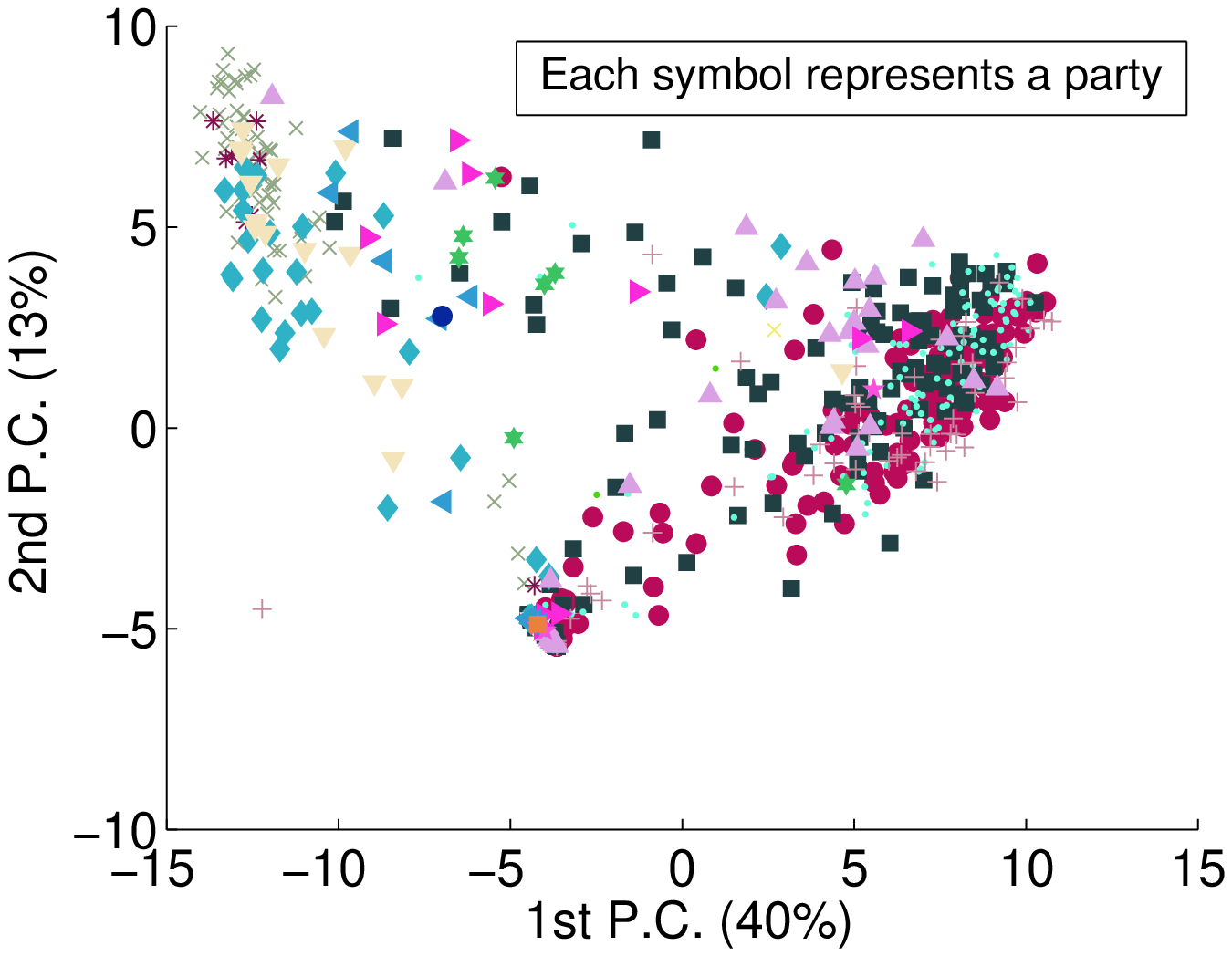}}
  \caption{The first two principal components of the PCA run for partisans' votes in the USA and in Brazil.}
  \label{fig:PCA}
\end{figure}

\section{The ARRANGE Method}
\label{sec:method}

In this section I describe the method \method, which has, basically, two steps. First, based on the votes given by party leaders, it tries to find pairs of parties that can be merged into one. The idea is that parties that always give the same vote could be merged into a single party. Then, \method attempts to assign partisans to new parties with the objective of minimizing the total number of parties receiving partisans and preserving party discipline. Finally, I describe the quality outputs of \method. Formally, the problem tackled by \method is:
 
\begin{problem}
Given a set of parties $\mathcal{P}$, a set of partisans $\mathcal{A}$, a set of bills and amendments $\mathcal{B}$ and the set of all votes given by parties and partisans $\mathcal{V}$, find the minimum set of parties $\mathcal{P}^*$ to which the partisans in $\mathcal{A}$ can be assigned in a way that overall, party and partisan disciplines are maximized. 
\label{problem1}
\end{problem}

\subsection{Merging Political Parties}

In order to assign partisans to other parties, it is necessary to formally define the ways it can be done. Thus, here I define an \textbf{option} as the descriptor of the party that can receive an external partisan as member. More formally, given a partisan $a:=(u,p_i)$, his/her current party $p_i$, and his/her set of propositions $\mathcal{B}_a$, an option $o_a := (a, p_j, sim(a,p_j))$ is a tuple composed by the partisan $a$, a party $p_j \neq p_i$, and the similarity value $sim(a,p_j)$ between $a$ and $p_j$. More importantly, the option $o_a = (a, p_j, sim(a,p_j))$ exists \textit{if and only if} $\mathcal{B}_a \subseteq \mathcal{B}_{p_j}$, i.e., if the party $p_j$ has voted for all propositions in $\mathcal{B}_a$. The set $\mathcal{O}_a=\{o_a^1, o_a^2, ...\}$ is composed by all the options of partisan $a$ or, more formally
\begin{IEEEeqnarray}{rCl}
	\mathcal{O}_a &=& \{[o_a := (a:=(u,p_i), p_j, sim(a,p_j))] ~: \nonumber\\ && p_i,p_j \in \mathcal{P} ~\wedge~ p_j \neq p_i ~\wedge~ \mathcal{B}_a \subseteq \mathcal{B}_{p_j}\}.
\end{IEEEeqnarray}
Moreover, an option $o_a := (a, p_j, sim(a,p_j))$ of partisan $a:=(u,p_i)$ is characterized as a \textit{good option} if $sim(a,p_j) \geq sim(a,p_i)$. The set of good options $\mathcal{O}_a^*$ for partisan $a:=(u,p_i)$ is defined as 
\begin{IEEEeqnarray}{rCl}
\mathcal{O}_a^* &=& \{[o_a = (a:=(u,p_i), p_j, sim(a,p_j))] \in \mathcal{O}_a ~: \nonumber\\ && sim(a,p_j) \geq sim(a,p_i)\} \label{eq:goodoptions}
\end{IEEEeqnarray}

In Fig.~\ref{fig:numberofoptions}a, I show the histogram of the number of good options $|O_a^*|$ for all partisans of our dataset. Consistent with the high levels of party and partisan discipline, observe that the majority of partisans do not have a single good option, i.e., for $69\%$ of the partisans there is no other party in Brazil that offers more similar voting vectors. In fact, only \pedror{$\approx 11\%$} of the partisans have more than three good options. This result suggests that it is very difficult to reduce the number of parties in Brazil by moving partisans from one party to another.

\begin{figure}[!hbt]
\centering
  \subfigure[Original parties]
            {\includegraphics[width=.45\textwidth]{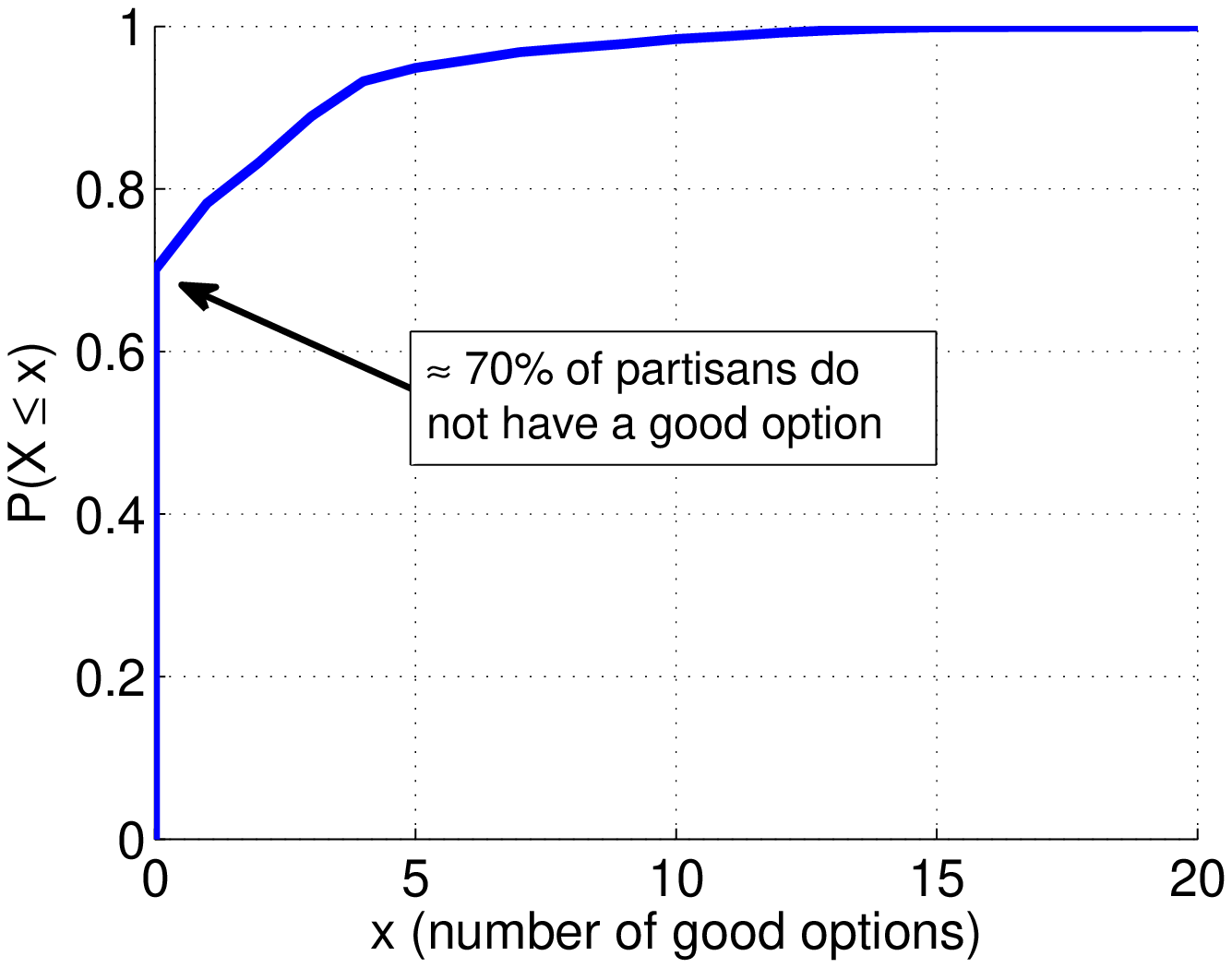}}
	\subfigure[Merged parties]
            {\includegraphics[width=.45\textwidth]{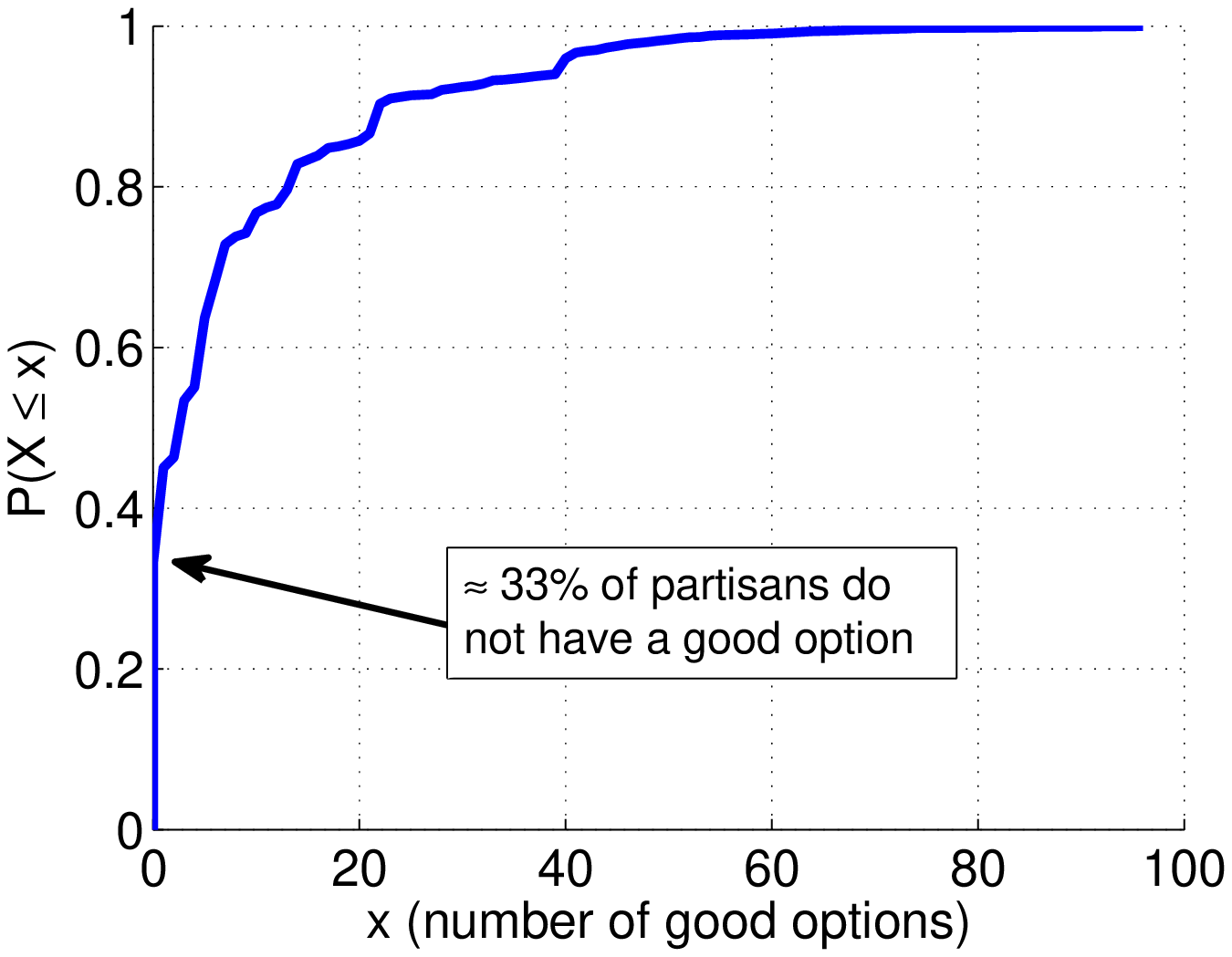}}
  \caption{Number of good options per partisan.}
  \label{fig:numberofoptions}
\end{figure}

The main reason for the low number of good options is related to the short lifetime of many parties in Brazil, as it can be observed in Fig.~\ref{fig:partysize}. Since an option exists if and only if the propositions voted on by the partisan is a subset of the propositions voted on by the parties, many partisans with a long history of votes cannot find options for them among the small parties. Thus, here I propose a method for creating new parties by merging existing ones. The method is based on a simple idea: if two parties are not contemporary or are contemporary, but all votes given by them are equal, then these parties can be merged into a new one.

More formally, two parties $p_i$ and $p_j$ can be merged into a new party $p_{i\_j}$ if one of the two conditions below is met:
\begin{itemize}
	\item \textbf{C1.} $\mathcal{B}_{p_i} \cap \mathcal{B}_{p_j} = \emptyset$, i.e., parties $p_i$ and $p_j$ have not voted for a common proposition.
	\item \textbf{C2.} $\forall b \in \mathcal{B}_{p_i} \cap \mathcal{B}_{p_j}: v_{p_i}^b = v_{p_j}^b$, i.e., parties $p_i$ and $p_j$ have given the same vote for all common propositions.
\end{itemize}

Given these two conditions, the first thing we have to do is find all pairs of parties to which at least one of the two conditions is valid. Once this is done, for every pair of parties $(p,q)$ that can be merged, we create a merged party $p\_q$, for which the sets of propositions and votes are, respectively,  $\mathcal{B}_{p\_q} = \mathcal{B}_{p} \cup \mathcal{B}_{q}$ and $\mathcal{V}_{p\_q} = \mathcal{V}_{p} \cup \mathcal{V}_{q}$. All parties created in this step are put in set $\mathcal{P}'$. After this, we repeat this process by verifying, for each party $p \in \mathcal{P}'$, all parties $q \in \mathcal{P}$ that can be merged to $p$. We merge $p$ to $q$ into $p\_q$ as previously, but this time adding the merged parties to $\mathcal{P}''$. Once this process is done, we copy $\mathcal{P}''$ to $\mathcal{P}'$, empty $\mathcal{P}''$, and restart the process of finding parties $q \in \mathcal{P}$ eligible to be merged to parties $p \in \mathcal{P}'$. The process ends when set $\mathcal{P}''$ does not receive a new merged party. All parties $p \in \mathcal{P}'$ that were not merged are put in the final set of merged parties $p \in \mathcal{P}^M$. This whole process is described in Algorithm~\ref{alg:partymerger}.

\begin{algorithm}
\scriptsize
\caption{Creates a set $\mathcal{P}^M$ of merged parties.}\label{alg:partymerger}
\begin{algorithmic}[1]
\Procedure{merge}{$p,q$}\Comment{two parties to be merged}
\State $\mathcal{B}_{p\_q} \gets \mathcal{B}_p \cup \mathcal{B}_q$
\State $\mathcal{V}_{p\_q} \gets \mathcal{V}_p \cup \mathcal{V}_q$
\Return \textbf{new} \texttt{party} $p\_q$
\EndProcedure
\Statex
\Procedure{Merge All Parties}{$\mathcal{P}$}\Comment{the party set $\mathcal{P}$}
\State $\mathcal{P}^M \gets \{\}$ \Comment{the set of merged parties}
\State $\mathcal{P}' \gets \mathcal{P}$ \Comment{merged parties to investigate}
\While{$\mathcal{P}'\not=\emptyset$} 
\State $\mathcal{P}''= \{\}$ \Comment{new merged parties}
\ForAll{$p \in \mathcal{P}'$}
\State $merged\gets$ \texttt{False}
\ForAll{$q \in \mathcal{P}$}
\If{$(\mathcal{B}_{p} \cap \mathcal{B}_{q} = \emptyset) \vee (\forall b \in \mathcal{B}_{p} \cap \mathcal{B}_{q}: v_{p}^b = v_{q}^b)$ }
\State $p\_q\gets \Call{merge}{p,q}$ 
\State $merged\gets$ \texttt{True}
\State $\mathcal{P}'' \gets \mathcal{P}'' \cup \{p\_q\}$
\EndIf
\EndFor
\If{$merged = $\texttt{False} } \Comment{cannot merge another party to $p$}
\State $\mathcal{P}^M \gets \mathcal{P}^M \cup \{p\}$
\EndIf
\EndFor
\State $\mathcal{P}' \gets \mathcal{P}''$
\EndWhile
\Return $\mathcal{P}^M$
\EndProcedure
\end{algorithmic}
\end{algorithm}

Basically, what Algorithm~\ref{alg:partymerger} does is to find all possible $k-$com\-bi\-na\-tions $\binom{N}{k}$ of set $\mathcal{P}$ for all $1 \leq k \leq N$. It is well known that $\sum_{k=1}^{N} \binom{N}{k} = 2^N-1$, making the worst-case complexity for this problem to be $O(2^N)$. Nevertheless, in practice, finding parties in $\mathcal{P}$ that can be merged with parties in $\mathcal{P}'$ gets significantly harder as $k$ increases, which, in practice, makes Algorithm~\ref{alg:partymerger} computationally feasible for the problem in question. For the case of the 36 Brazilian parties, $k$ went up to 6 and the algorithm stopped, generating a set $\mathcal{P}^M$ containing 95 parties, in which only 5 were not a product of a merge, namely PDT, PMDB, PPS, PSDB and PT. Besides these, Algorithm~\ref{alg:partymerger} generated 12 parties of size 2, 17 parties of size 3, 20 parties of size 4, 31 parties of size 5 and 10 parties of size 6.

In Fig.~\ref{fig:numberofoptions}b, I show the histogram of the number of good options $|O_a^*|$ for all partisans considering the new set of merged parties $\mathcal{P}^M$. Observe that the number of partisans that do not have a single good option dropped from $\approx 69\%$ to $\approx 33\%$, all of them being members of the five parties that were not merged. Moreover, the number of partisans that have more than three good options grew from $\approx 12\%$ to \pedror{$\approx 47\%$}. Thus, I conclude that merging parties significantly raises the chances of reducing the number of parties in Brazil by moving partisans from one party to another without decreasing party discipline.

\subsection{Finding the Minimum Set}

Now that most of the partisans have multiple good options, we can find ways of redistributing them among the parties in $\mathcal{P}^M$. The idea in this redistribution is to find a set of parties $\mathcal{P}^* \subset \mathcal{P}^M$ that has a lower cardinality than $\mathcal{P}$, i.e., $\mathcal{P}^* \subset \mathcal{P}^M$ is a set of parties able to receive all partisans $a \in \mathcal{A}$ as members and the number of parties $N^* = |\mathcal{P}^*|$ in $\mathcal{P}^*$ has to be lower than the actual number of parties $N = |\mathcal{P}|$. However, this has to be done cautiously, since there are two partially conflicting goals:

\begin{enumerate}
	\item Minimize the number of parties;
	\item Maximize party and partisan discipline.
\end{enumerate}

These goals are conflicting because minimizing the number of eligible parties to receive partisans as members implies reducing the options for moving partisans and, as a consequence, the number of good options. If a partisan does not have a good option, he/she is obliged to stay in his/her party in order to not decrease his/her partisan discipline. On the other hand, maximizing the discipline implies in maximizing the size of the set of good options and, therefore,  the number of parties to be considered has to be as high as possible.

In order to solve this conflict, I model this redistribution problem as a \textbf{set cover problem} (\textbf{SCP})~\cite{Chvatal1979}. SCP is a well studied problem for the field of approximation algorithms~\cite{Alon2006}, being also one of Karp's 21 NP-complete problems shown to be NP-complete. In summary, given a set of elements $\{1,2,...,m\}$ (called the universe) and a set $S$ of $n$ sets whose union equals the universe, the SCP is to identify the smallest subset of $S$ whose union equals the universe. More formally, given a universe $\mathcal{X}$ and a family $\mathcal{S}$ of subsets of $\mathcal{X}$, a \textit{cover} is a subfamily $\mathcal{C}\subseteq\mathcal{S}$ of sets whose union is $\mathcal{X}$.

In our case, the universe $\mathcal{X}$ is the set of partisans $\mathcal{A}$ and the family $\mathcal{S}$ of subsets of $\mathcal{X}$ is a family of subsets of partisans $\mathcal{A}_p^M \in \mathcal{A}$ where each subset  $\mathcal{A}_p^M$ is composed by the partisans that are eligible for moving to party $p \in \mathcal{P}^M$. In order to build $\mathcal{A}_p^M$, it is necessary to recalculate the set of options $\mathcal{O}_a$ and good options $\mathcal{O}_a^*$ of each partisan $a$ with respect to the merged parties in $\mathcal{P}^M$. Once this is done, I define that each set $\mathcal{A}_p^M$ is composed by all the partisans $a$ that have a good option $o_a := (a, p, sim(a,p))$, i.e., $\mathcal{A}_p^M = \{a | a \in \mathcal{A} \wedge p \in \mathcal{P}^M \wedge \exists ~ [o_a := (a, p, sim(a,p))] \in \mathcal{O}_a^* \}$.

Since this problem is NP-complete, I recur to a greedy algorithm to solve it. Literature shows that the greedy algorithm is essentially the best possible polynomial time approximation algorithm for set cover under plausible complexity assumptions~\cite{Feige1998}. The greedy algorithm for set covering chooses sets according to one rule: at each stage, choose the set that contains the largest number of uncovered elements. It can be shown~\cite{Chvatal1979} that this algorithm achieves an approximation ratio of $H(s)$, where $s$ is the size of the set to be covered and $H(n)$ is the $n$-th harmonic number: $H(n) = \sum_{k=1}^{n} \frac{1}{k} \le \ln{n} +1$.

For our specific problem, this algorithm works as follows, being described in Algorithm~\ref{alg:greedy}. First, we create two empty sets: $\mathcal{P}^*$, which will contain the final set of parties, and $\mathcal{A}'$, which receives the covered partisans during the process. Then, while $\mathcal{A}'$ does not contain all partisans, we find the party $p \in \mathcal{P}^M$ to which $\mathcal{A}_p^M$ contains the largest number of uncovered partisans. Then, we add all partisans $a \in \mathcal{A}_p^M$ to $\mathcal{A}'$, make $\mathcal{A}_p^M$ an empty set (so it is not selected again), and add $p$ to the final set of parties $\mathcal{P}^*$. 

Note that so far this process \textit{guarantees} that all partisans will, at least, have the same partisan discipline as their actual ones, since only good options are used to build the sets $\mathcal{A}_p^M, p \in \mathcal{P}^M$. Nevertheless, it is possible that relaxing this constraint a little might diminish the total number of parties that compose $\mathcal{P}^*$ considerably. For instance, we may allow a partisan to be member of a party if their similarity is at most $0.05$ smaller than his/her similarity with his/her actual party. This relaxation increases the number of partisans that are eligible to be member of other parties and, therefore, may reduce the size of $\mathcal{P}^*$.

Thus, here I introduce the parameter $\delta$, which is the maximum allowed difference between the actual partisan discipline and the future one. We accommodate this in \method by simply changing the way the set of good options is constructed. With the introduction of $\delta$, the set of good options $\mathcal{O}_a^*$ is defined as:

\begin{IEEEeqnarray}{rCl}
\mathcal{O}_a^* &=& \{[o_a := (a:=(u,p_i), p_j, sim(a,p_j))] \in \mathcal{O}_a ~:~  p_i \in \mathcal{P} ~\wedge\nonumber\\ && p_j \in \mathcal{P}^M ~\wedge~ sim(a,p_j) \geq sim(a,p_i) - \delta\}. \label{eq:goodoptionsdelta}
\end{IEEEeqnarray}

In summary, after generating the set of merged parties $\mathcal{P}^M$ using Algorithm~\ref{alg:partymerger}, it is necessary to create the sets $\mathcal{A}_p^M$ for all $p \in \mathcal{P}^M$ considering $\delta$. These sets $\mathcal{A}_p^M$ will contain all partisans that are eligible for being members of party $p$ given $\delta$. This is done by selecting a value for $\delta$ and then running the procedure \Call{find}{} of Algorithm~\ref{alg:greedy}. If, for instance, $\delta=0$, then no partisans are allowed to decrease their discipline when moving to a different party. On the other side, if $\delta=1$, partisans are allowed to be members of any party that voted for all propositions that they voted on, i.e., partisan discipline is not a constraint. After running Algorithm~\ref{alg:greedy}, we will have a minimum set of parties $\mathcal{P}^*$ able to accommodate all partisans in $\mathcal{A}$. Then, all we have to do is to assign a party $p \in \mathcal{P}^*$ to each partisan $a \in \mathcal{A}$. This is done by selecting the option $o_a = (a:=(u,p_i), p_j, sim(a,p_j)) \in {O}_a^*$ that gives the maximum similarity value $sim(a,p_j)$ for all parties $p_j \in \mathcal{P}^*$ and making $p_j$ the new party of partisan $a$.

\begin{algorithm}
\scriptsize
\caption{Find the minimum set of parties $\mathcal{P}^*$}\label{alg:greedy}
\begin{algorithmic}[1]
\Procedure{buildOptions}{$\mathcal{A}_p^M, \delta$}
\ForAll{$a \in \mathcal{A}$}
\State $\mathcal{O}_a \gets \{[o_a := (a:=(u,p), q, sim(a,q))] ~:~ p \in \mathcal{P} ~\wedge~ q \in \mathcal{P}^M ~\wedge~ \mathcal{B}_a \subseteq \mathcal{B}_{q}\}$
\State $\mathcal{O}_a^* \gets \{[o_a := (a:=(u,p), q, sim(a,q))] \in \mathcal{O}_a ~:~ sim(a,q) \geq sim(a,p)-\delta\}$
\EndFor
\EndProcedure
\Statex
\Procedure{find $\mathcal{P}^*$}{$\mathcal{A}_p^M, \delta$}
\State \Call{buildOptions}{$\mathcal{A}_p^M, \delta$}
\State $\mathcal{A}' \gets \{\}$ \Comment{the set of uncovered partisans}
\State $\mathcal{P}^* \gets \{\}$ \Comment{the final set of parties}
\While{$\mathcal{A}' \not= \mathcal{A}$} 
\State $\mathcal{P}''\gets \{\}$ \Comment{new merged parties}
\State $p \gets \operatorname*{arg\,max}_p |\mathcal{A}_p^M \cup \mathcal{A}|$
\State $\mathcal{A}' \gets \mathcal{A}' \cup \mathcal{A}_p^M$
\State $\mathcal{A}_p^M \gets \{\}$
\State $\mathcal{P}^* \gets \mathcal{P}^* \cup \{p\}$ 
\EndWhile
\Return $\mathcal{P}^*$
\EndProcedure
\end{algorithmic}
\end{algorithm}

\subsection{Quality Signals}

The main goal of \method is to reduce party fragmentation by reducing the number of parties that are able to accommodate all elected partisans. Thus, the main output of \method is $N^*$ (or $|\mathcal{P}^*|$). Nevertheless, it is necessary to assess the quality of the \textit{cover} $\mathcal{P}^*$, since it is essential to reduce party fragmentation by achieving desirable levels of party and partisan discipline. But what are the desirable levels of party and partisan discipline?

In Fig.~\ref{fig:partisandiscipline}a, I showed the CDF of the actual partisan discipline distribution during the analyzed period. The most desirable, or \textit{ideal}, partisan discipline distribution is shown as a dashed red line, which represents the situation where all partisans have discipline of $1.0$. Thus, a new partisan discipline distribution, generated by Algorithm~\ref{alg:greedy} with parameter $\delta$ and defined by the random variable $X_{\delta}$, is considered \textit{desirable} if its CDF $F_{X_{\delta}}(x)$ is closer to the ideal than the actual one, defined by the random variable $X_0$ and CDF $F_{X_0}(x)$. More formally, considering that the area under the ideal CDF curve is $0$ for both partisan and party discipline distributions, I propose the following definition:
\begin{definition}
A new discipline distribution defined by random variable $X_{\delta}$ and its CDF $F_{X_{\delta}}(x)$ is considered a \textit{\textbf{desirable discipline distribution}} if $\int_{0}^{1}F_{X_0}(x)~dx - \int_{0}^{1}F_{X_{\delta}}(x)~dx > 0$, where $F_{X_0}(x)$ is the CDF of the actual discipline distribution defined by random variable $X_0$. \label{def:desirable}
\end{definition}

In other words, if the area under $F_{X_{\delta}}(x)$ is smaller than the area under $F_{X_0}(x)$, then $X_{\delta}$ represents a \textit{desirable discipline distribution}. Moreover, given that discipline ranges from $0$ to $1$, I propose the following lemma:

\begin{lemma}
Given a random variable $X_0$ representing the actual discipline distribution and a random variable $X_{\delta}$ representing a new discipline distribution, if the expected value $E_{X_{\delta}}[x]$ of $X_{\delta}$ is higher than the expected value $E_{X_0}[x]$ of $X_0$, then $X_{\delta}$ represents a \textit{desirable discipline distribution}.
\end{lemma}
\begin{proof}
From probability theory, we know that $E_X[x] = \int_{0}^{\infty}1-F_X(x)~dx$ for a given random variable $X$ and CDF $F_X(x)$~\cite{billingsley1986probability}. Since discipline has values from $0$ to $1$, we can write $E_X[x] = 1-\int_{0}^{1}F_X(x)~dx$, or $\int_{0}^{1}F_X(x)~dx = 1-E_X[x]$ for discipline distributions. Then, we can replace Definition~\ref{def:desirable} for $1-E_{X_{\delta}}[x] < 1-E_{X_0}[x]$ or $E_{X_{\delta}}[x] > E_{X_0}[x]$. 
\end{proof}

Thus, now I can formally define three binary quality signals for the new partisan configuration over the new set of parties $\mathcal{P}^*$ generated by Algorithm~\ref{alg:greedy} with parameter $\delta$. These signals indicate, respectively, if the new partisan configuration has desirable levels of partisan, party and overall discipline, and is defined as:

\begin{itemize}
	\item \textbf{$Q_1$}: $1$ if the overall discipline $d_*^{\delta}$ of the new configuration is greater than the overall discipline $d_*^{0}$ of the actual configuration; $0$ otherwise. Recall that the overall discipline $d_*$ is defined in Equation~\ref{eq:overaldiscipline}.
	\item \textbf{$Q_2$}: $1$ if the expected (average) partisan discipline of the new configuration $E_{X_{\delta}}^{\mathcal{A}}[x]$ is greater than the expected party discipline of the actual configuration $E_{X_{0}}^{\mathcal{A}}[x]$ or, in other words, if the new partisan discipline distribution is a \textit{desirable partisan discipline distribution}; $0$ otherwise.		
	\item \textbf{$Q_3$}: $1$ if the expected (average) party discipline of the new configuration $E_{X_{\delta}}^{\mathcal{P}}[x]$ is greater than the expected party discipline of the actual configuration $E_{X_{0}}^{\mathcal{P}}[x]$ or, in other words, if the new party discipline distribution is a \textit{desirable party discipline distribution}; $0$ otherwise.
\end{itemize}

\section{Results}

In this section I show the results for \method for a hundred equally distributed values of $\delta$ between $0$ and $1$. From now on I will call a configuration $c_{\delta}$ the distribution of partisans among parties generated by \method for a particular $\delta$ value. Moreover, in all plots I will indicate whether the quality signals described in previous section were $1$ or $0$. Namely, I will use yellow stars for results where all quality signals were $1$ ($Q_1 \wedge Q_2 \wedge Q_3$), blue diamonds when only signals $Q_2$ and $Q_3$ were $1$ ($Q_2 \wedge Q_3$), red circles when only signal $Q_3$ was $1$ and, finally, green squares when all quality signal were $0$. In summary, \method was able to generate $31$ distinct configurations that, compared with the \statusquo, have (i) a significantly smaller number of parties, (ii) higher discipline of partisans towards their parties and (iii) more even distributions of partisans into parties. Besides comparing with the \statusquo, \method will be compared with two random models: \randomsq and \randomdelta. The competitors can be summarized as:

\begin{enumerate}
	\item \statusquo. It is the existing state of affairs, i.e., the actual and historical situation in Brazil during the analyzed period.
	\item \randomsq. It randomly redistributes the partisans among the existing parties. For each partisan $a \in \mathcal{A}$, the model randomly pick an option $[o_a := (a, p, sim)] \in \mathcal{O}_a$ and assigns $a$ to party $p$. By using the options set $\mathcal{O}_a$ I guarantee that the party $p$ allocated to $a$ has voted for every proposition voted by $a$, i.e., $\mathcal{B}_a \subseteq \mathcal{B}_p$.
	\item \randomdelta. Works in the same way as \randomsq, but instead of allocating partisans to the actual set of parties $\mathcal{P}$, it randomly redistributes the partisans among the minimum set of parties $\mathcal{P}^*$ generated by Algorithm~\ref{alg:greedy}. 
\end{enumerate}

The random models are used to quantify the payoffs obtained by using \method when it generates the minimum party set $\mathcal{P}^*$ (\randomsq) and it efficiently allocates partisans to the parties of this set (\randomdelta). Although not always visible, for all results of both models I also show the $99\%$ confidence interval.

In Fig.~\ref{fig:numberofparties}, I show the number of parties $N^*$ generated by \method for different values of $\delta$. I also show the \statusquo, i.e., the actual number of parties $N$ of which partisans were members during the analyzed period, and the number of parties generated by \randomsq. First, note that the number of parties generated by \method is significantly lower than the \statusquo, decreasing as $\delta$ increases. Even when no partisans are allowed to decrease its discipline ($\delta=0$), $N^*=22$, a number $\approx 39\%$ lower than $36$, the \statusquo. Moreover, \method was able to generate a configuration ($c_{0.15}$) with only \pedror{$13$ parties ($\approx 64\%$ reduction)} and with all quality signals equal to $1$ ($Q_1 \wedge Q_2 \wedge Q_3$), i.e., with overall, party and partisan disciplines greater than the \statusquo. When only $Q_2$ and $Q_3$ are $1$ ($Q_2 \wedge Q_3$), \method could generate a configuration \pedror{($c_{0.19}$) with $N^*=10$ parties, a $\approx 72\%$ reduction}. Finally, when only $Q_3$ is $1$ ($Q_3$), \method could provide a \pedror{$\approx 86\%$ reduction in the number of parties by generating a configuration (e.g. $c_{0.41}$) with only $5$ parties}. It is worth mentioning that the number of parties in Brazil is so excessive that even \randomsq was able to generate a configuration with fewer parties than the \statusquo.

\begin{figure}[!hbt]
\centering
            {\includegraphics[width=.65\textwidth]{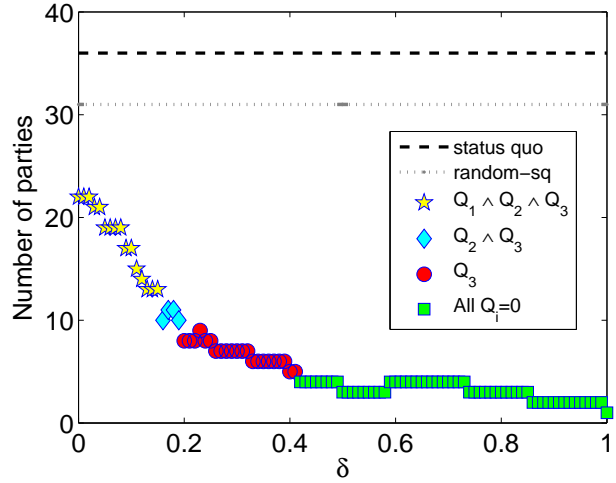}}
  \caption{Number of parties generated by \method.}	  
  \label{fig:numberofparties}
\end{figure}

Concerning discipline, I show in Fig.~\ref{fig:discipline} the overall discipline (Fig.~\ref{fig:discipline}a) and the average party (Fig.~\ref{fig:discipline}b) and partisan (Fig.~\ref{fig:discipline}c) disciplines of the configurations produced by \method and its competitors. Observe that for all discipline metrics the results produced by \method are very similar with the \statusquo, even when all quality signals are $0$. The discipline values decrease significantly only for $\delta$ values close to $1$, when the number of parties generated by \method is $2$ or $1$. It is also interesting to note that the random models are able to produce configurations with considerably high levels of discipline, which shows that Brazilian political parties are, on average, very similar to each other.

\begin{figure*}[!hbt]
\centering
  \subfigure[Overall discipline]
            {\includegraphics[width=.32\textwidth]{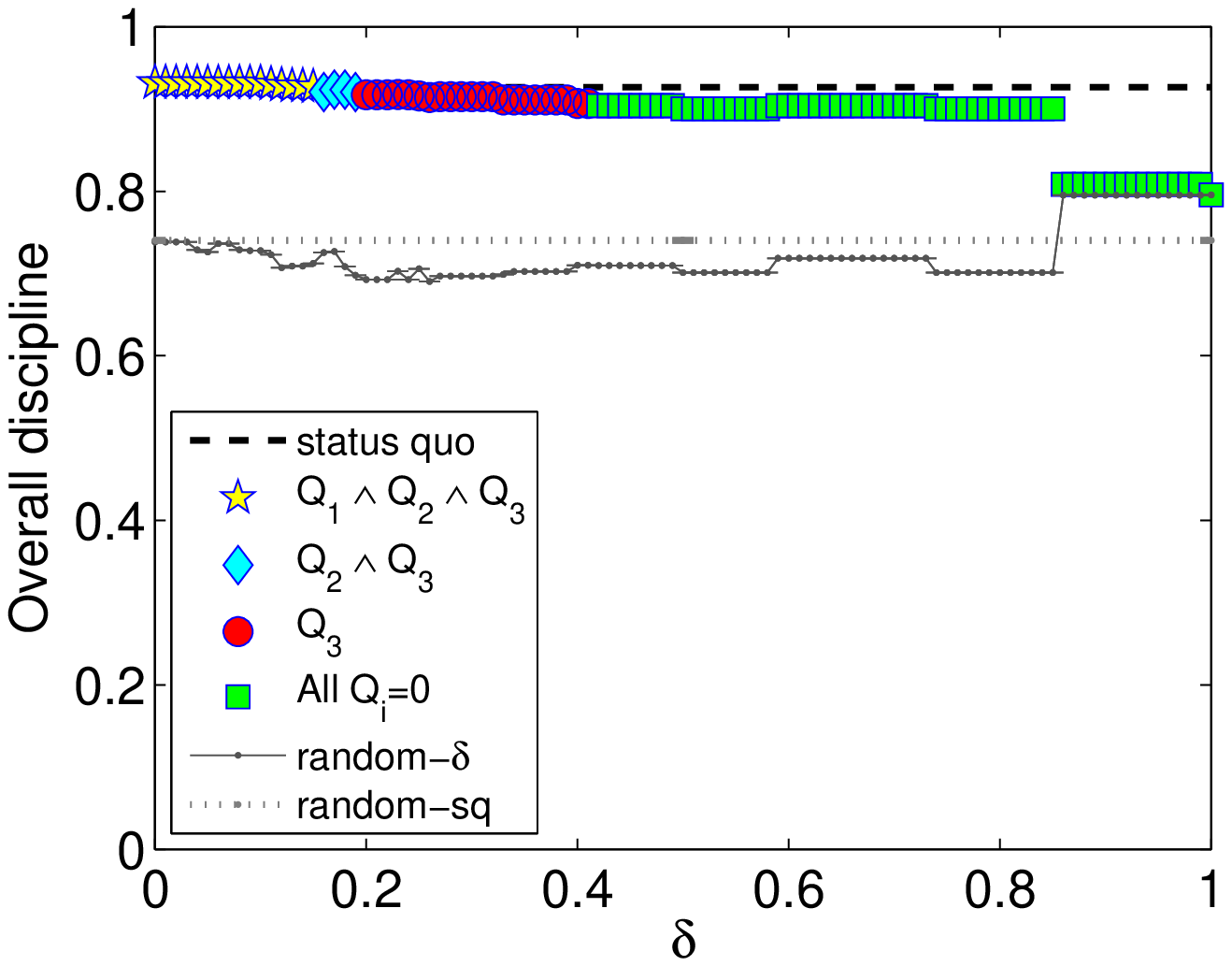}}
  \subfigure[Average party discipline]
            {\includegraphics[width=.32\textwidth]{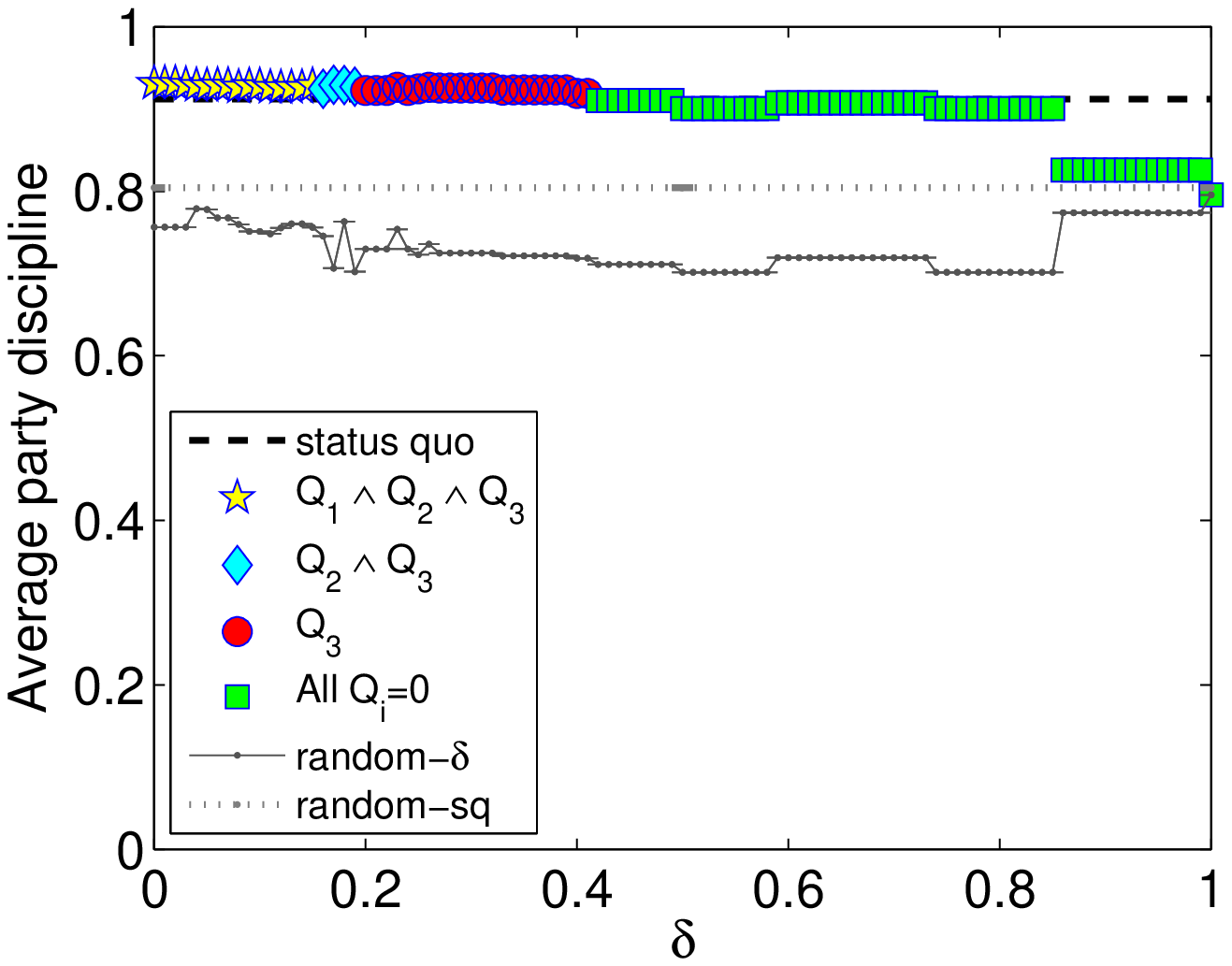}}
  \subfigure[Average partisan discipline]
            {\includegraphics[width=.32\textwidth]{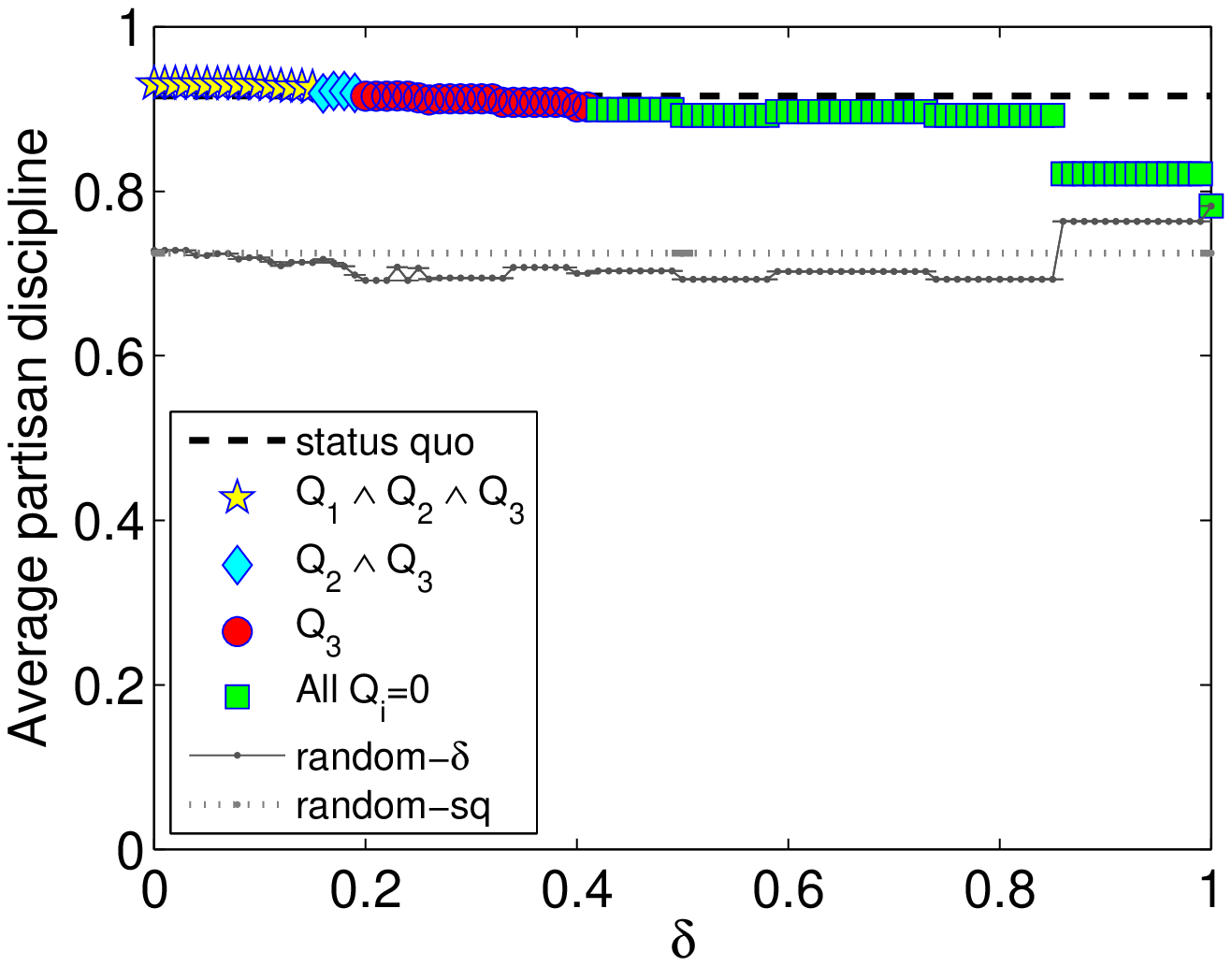}}												
  \caption{Overall and average party and partisan discipline generated by \method for different values of $\delta$ in comparison with the \statusquo and the random models \randomsq and \randomdelta.}  
  \label{fig:discipline}
\end{figure*}

In order to analyze how well distributed are the partisans among parties, we compute the Gini coefficient~\cite{Ceriani:gini:2012} for each configuration generated by \method and its competitors. The Gini coefficient was initially proposed to describe the income inequality in a population~\cite{Ceriani:gini:2012}. It assumes values from $0$, which expresses perfect equality, where all parties have the same number of partisans, to $1$, which expresses maximal inequality among values, where all partisans are allocated to a single party. Observe in Fig.~\ref{fig:ginicoefficient} that \method is able to produce configurations in which the partisans are more evenly distributed than the \statusquo for all values of $\delta$. While the Gini coefficient is $\approx 0.64$ for the \statusquo, it decreased to \pedror{$\approx 0.43$ for $c_{0.15}$} ($Q_1 \wedge Q_2 \wedge Q_3$), to \pedror{$\approx 0.32$ for configuration $c_{0.41}$ ($Q_3$)}, and to $\approx 0$ for \pedror{$c_{0.5}$} (all quality signals equal to $0$, but with discipline values similar with the \statusquo). Model \randomsq has a slightly high Gini coefficient than the \statusquo because a few big parties are more likely to randomly receive new partisans, since they appear as an option in the option sets $\mathcal{O}_a$ for all $a \in \mathcal{A}$.

\begin{figure}[!hbt]
\centering
	{\includegraphics[width=.65\textwidth]{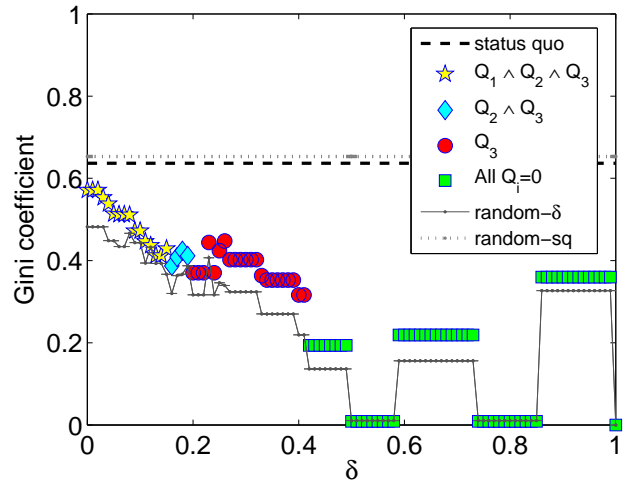}}
  \caption{Gini coefficient of the distribution of partisans among the parties.}
  \label{fig:ginicoefficient}
\end{figure}

Another relevant characteristic to be measured in political party systems is switching, i.e., party changes among partisans. As stated by~\cite{Desposato2006}, switching effectively destroys the meaning of party labels, raises voters' information costs, and eliminates party accountability, being a threat to the very core of democratic representation. Thus, in Fig.~\ref{fig:partychanges}, I plot the number of party changes among partisans that would occur if the configurations generated by \method and the random models were the reality. Observe that in all scenarios generated by \method the number of changes is lower than the \statusquo, \randomsq and \randomdelta. This is another evidence of the importance of reducing party fragmentation to have ideologically well defined parties.

\begin{figure}[!hbt]
\centering
	{\includegraphics[width=.65\textwidth]{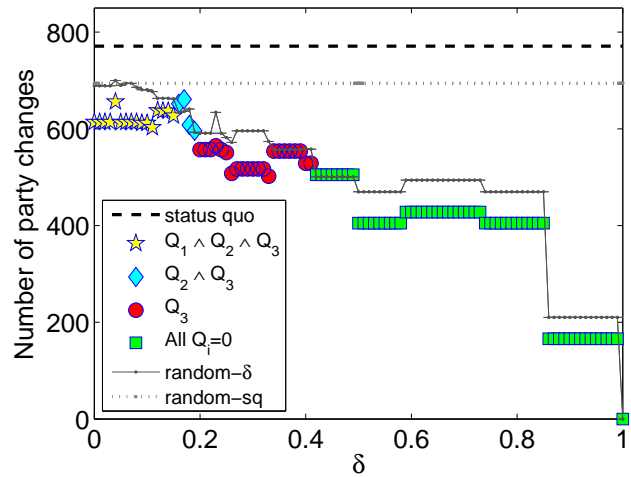}}
  \caption{Total number of party changes among partisans.}
  \label{fig:partychanges}
\end{figure}

Now I will take a closer look at particular configurations generated by \method, namely \pedror{$c_{0.15}$, $c_{0.19}$ and $c_{0.41}$}. These are the configurations that have the lowest number of parties and achieved, respectively, quality signals $Q_1 \wedge Q_2 \wedge Q_3$, $Q_2 \wedge Q_3$ and only $Q_3$. In Fig.~\ref{fig:npartiesperyear}a, I show the number of active parties per year for the \statusquo and these three configurations. First, observe that the number of parties changes constantly over the years in Brazil, this being a harmful consequence of its highly fragmented party system. On the other hand, observe that the configurations generated by \method are (i) much more stable over the years and (ii) have a significantly smaller number of parties than the \statusquo.

\begin{figure}[!hbt]
\centering
  \subfigure[Actual]
            {\includegraphics[width=.45\textwidth]{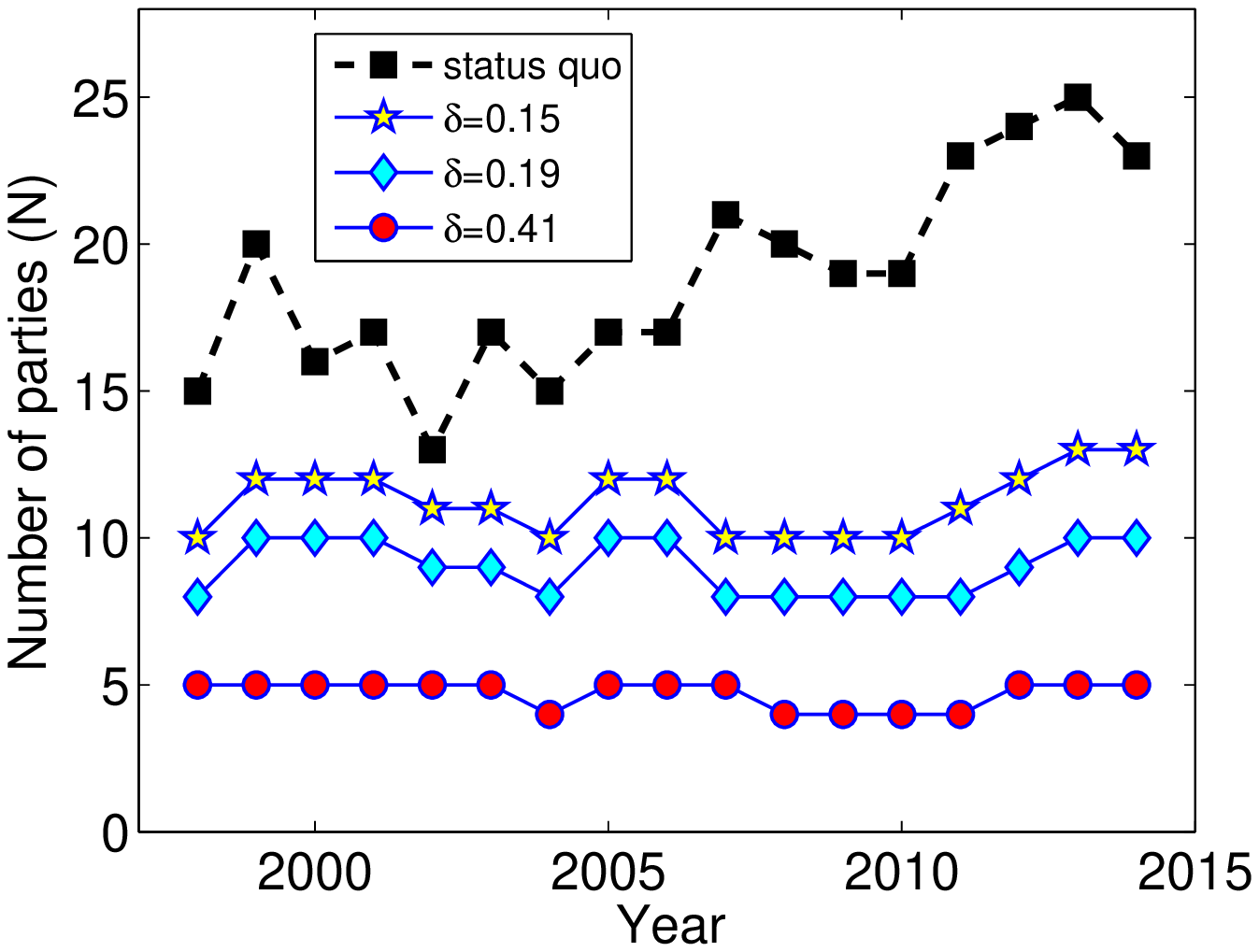}}
	\subfigure[Effective]
            {\includegraphics[width=.45\textwidth]{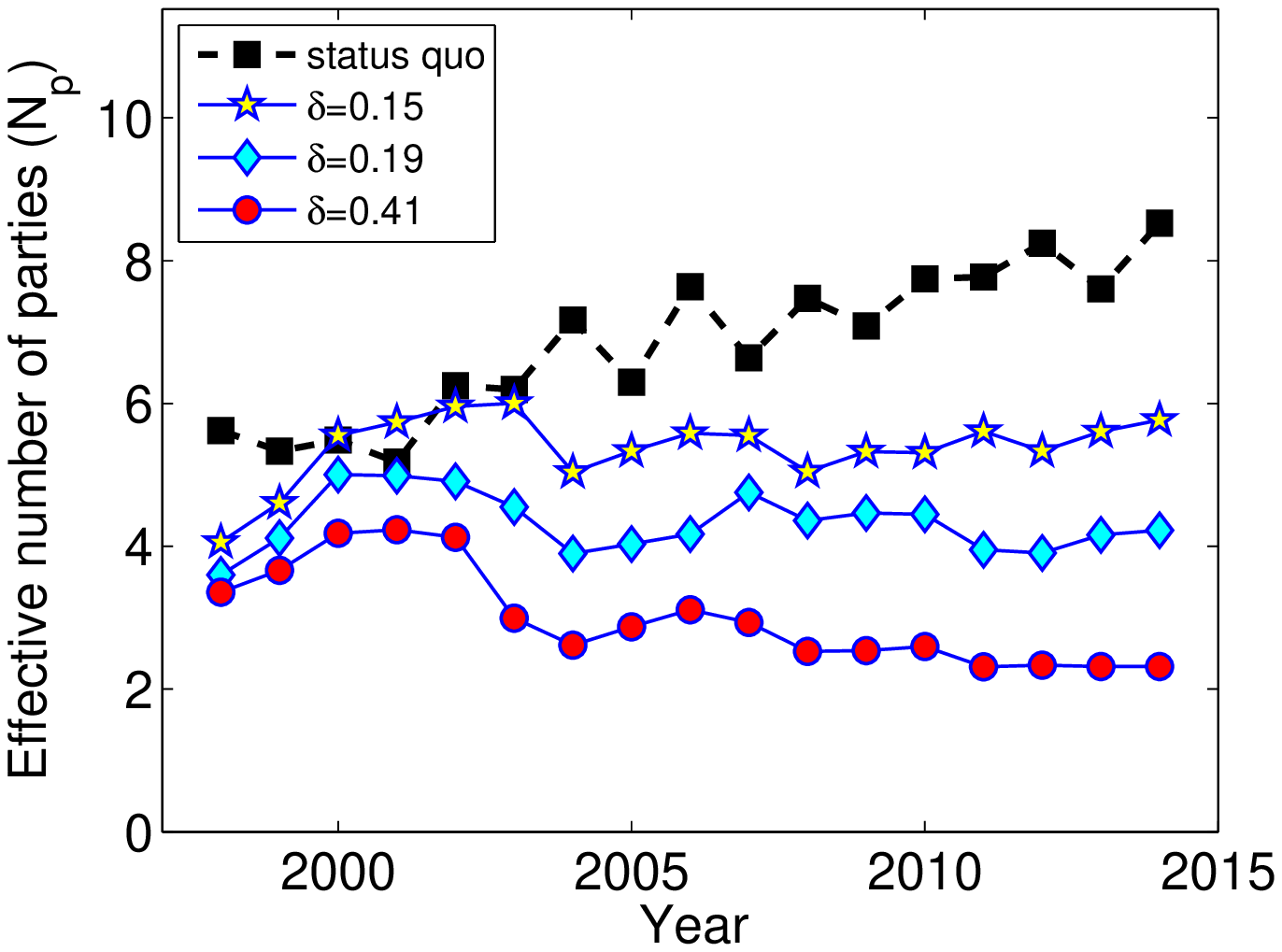}}
\caption{Number of parties per year.}
  \label{fig:npartiesperyear}
\end{figure}

Besides counting the actual number of parties per year, it is also important to measure the \textit {effective number of parties}. As described previously, it is a concept which provides for an adjusted number of political parties in a country's party system, weighting the partisan count per party by their relative strength~\cite{LaaksoTaagepera79}. In our case, the relative strength refers to their seat share in the parliament. This measure is especially useful to detect trends toward fewer or more numerous parties over time~\cite{LaaksoTaagepera79}. The number of parties equals the effective number of parties only when all parties have equal strength. In any other case, the effective number of parties is lower than the actual number of parties. It is also a frequent metric for the fragmentation of a party system~\cite{Anckar2000}. Moreover, although several indexes for computing the effective number of parties exist~\cite{Golosov2009}, in this paper I use the \textit{Golosov} index $N_p=\sum_{1}^{N}{(1+(s_1^2/s_i)-s_i)^{-1}}$, where $N$ is the actual number of parties, $s_i$ is the proportional share of each party $p_i$, and $s_1$ is the highest share of a party~\cite{Golosov2009}. For the best of my knowledge, $N_p$ is the most recent one and its results confirm it works better than earlier proposed alternatives in measuring the effective number of components in highly fragmented and highly concentrated party systems, which is the case of Brazil. 

Thus, in Fig.~\ref{fig:npartiesperyear}b, we show the \textit{effective number of parties} $N_p$ per year, calculated for configurations \pedror{$c_{0.15}$, $c_{0.19}$ and $c_{0.41}$} and for the \statusquo. First, observe that $N_p$ is significantly lower and more stable for \pedror{$c_{0.15}$, $c_{0.19}$ and $c_{0.41}$} than for the \statusquo. While $N_p$ grows constantly after the year $2000$ for the \statusquo, it remains practically constant for the three configurations generated by \method. Moreover, if we consider only \pedror{$c_{0.41}$}, Brazil would go from having one of the most fragmented party systems in the world~\cite{Limongi2009} to having one of the least fragmented~\cite{Anckar2000}, averaging \pedror{$3.0$} effective parties per year.

In Fig.~\ref{fig:PCA}b I showed the first two components of the PCA for the matrix $M_{\mathcal{V}_{BR}}$ composed by the votes of the partisans of Brazil during the years of 1999 and 2000. In Fig.~\ref{fig:PCAnew}, I show this same result, but I replace the \statusquo party labels by the ones generated by \method in \pedror{$c_{0.15}$, $c_{0.19}$ and $c_{0.41}$}. Observe that all three configurations have a visually better clustering of partisans of the same party than the \statusquo. This suggests that \method is also able to provide configurations in which parties are more different among themselves than in the \statusquo. I leave a deeper quantitative analysis for future work.

\begin{figure*}[!hbt]
\centering
  \subfigure[\pedror{$\delta=0.15$}]
            {\includegraphics[width=.32\textwidth]{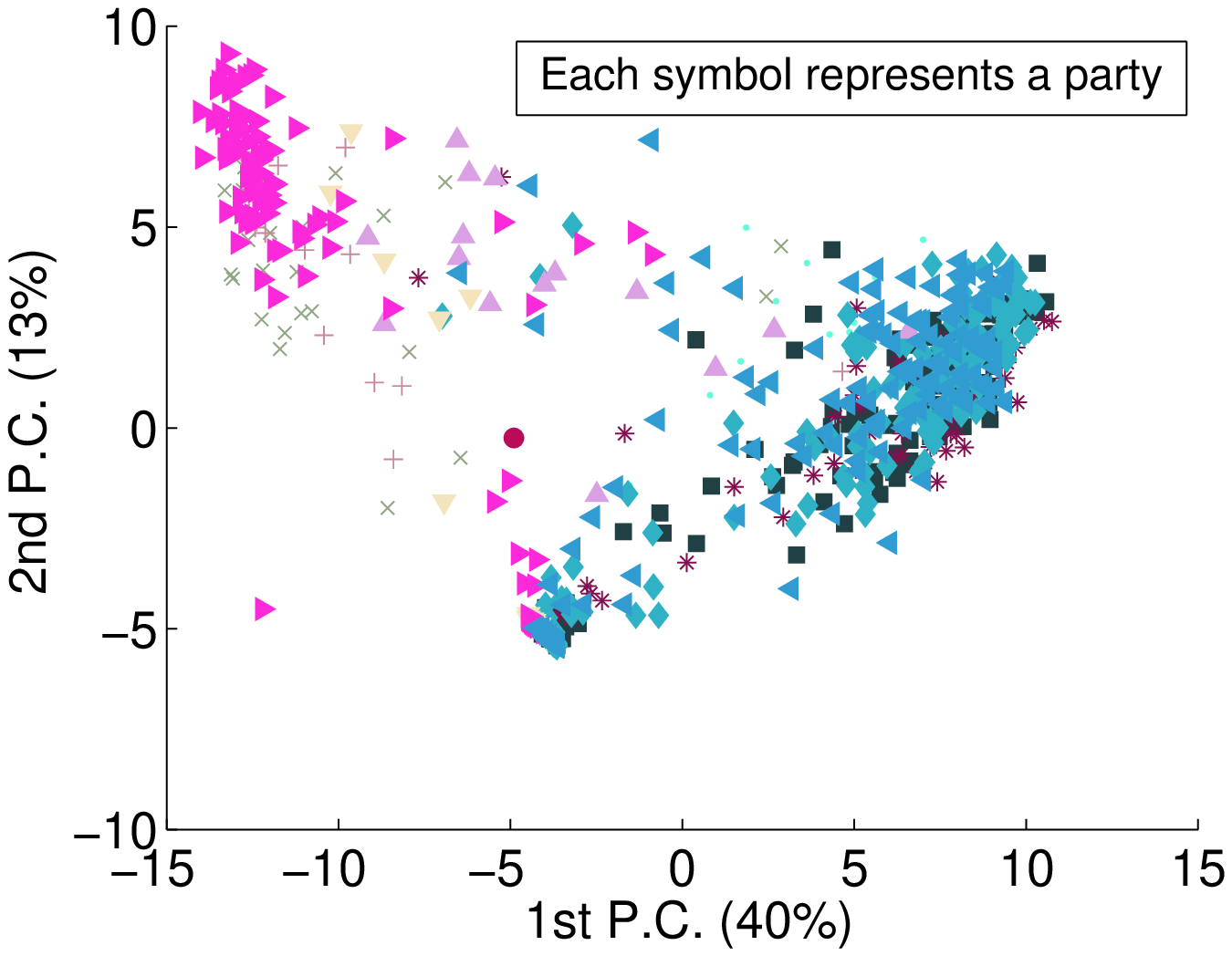}}
  \subfigure[\pedror{$\delta=0.19$}]
            {\includegraphics[width=.32\textwidth]{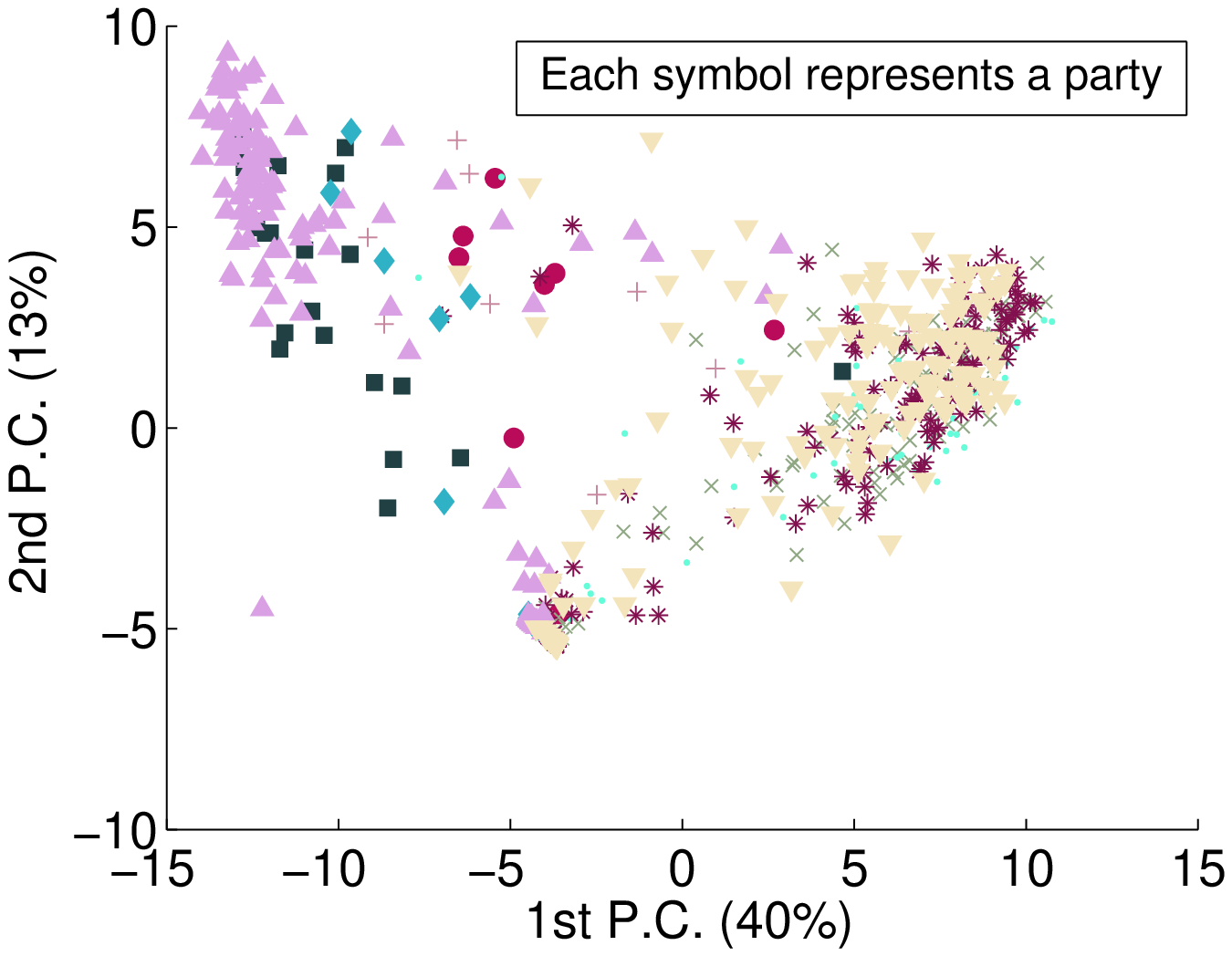}}
  \subfigure[\pedror{$\delta=0.41$}]
            {\includegraphics[width=.32\textwidth]{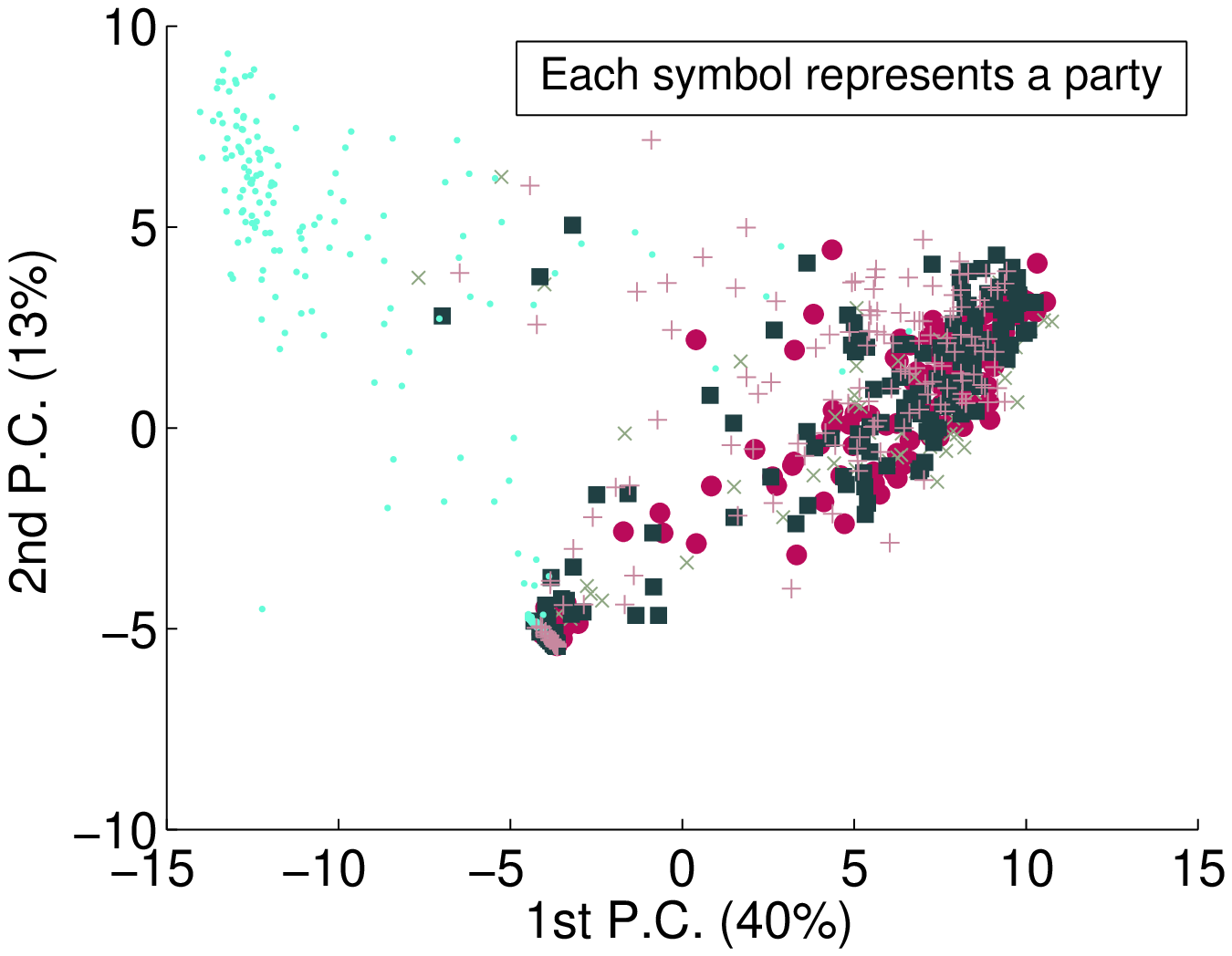}}												
  \caption{The first two principal components of the PCA run for partisans' votes in Brazil considering the redistribution of partisans performed by \method.}
  \label{fig:PCAnew}
\end{figure*}

\section{Discussion}

In practical terms, the main contribution of \method to a government and its population is the ability to provide a quantitative assessment of fragmentation in its party system. Given that the effective number of parties is expected to reflect the number of issue (or ideological) dimensions in a country~\cite{TAAGEPERA2006}, highly fragmented party systems overestimate this number of issue dimensions, providing a distorted view to the population that harms democracy~\cite{ranney1975curing}. Thus, by mapping plenary votes into ideological preferences, \method quantitatively provides an ``ideal'' number of parties for a country given the ideological preferences of its congressmen, revealing the \textit{presumable} true number of issue dimensions that exists in this particular country. I used the expression ``\textit{presumable
} true'' because, as verified by~\cite{Limongi2009}, it is not always true that a disciplined and cohesive party represents an ideological cleavage (or group) in society, i.e., its members may simply be a group of congressmen obeying its leader in order to obtain a particular benefit.

It is also worth mentioning that \method also provides the list of parties that should ideally exist and which partisans should be their members. Although I know a democratic government cannot implement this solution easily, it can be used to support significant reforms in its political system. With this in mind, I could not fail to mention in this paper that one of the possible reasons for Brazil's high level of party fragmentation is the so called \textit{fundo partid\'{a}rio}, which are funds distributed by the federal government to Brazilian political parties for them to spend indiscriminately. A share of the amount paid by the federal government through the \textit{fundo partid\'{a}rio} is the same for every party, but another part is proportional to the number of elected congressmen, senators and governors by each organization. In 2014, \textit{PT}, the party with the highest share ($16.5\%$), received R$\$50,314,999.19$ million \textit{reais} from the fund. On the other hand, \textit{PROS}, the party with the lowest share ($0.16\%$), received R\$493,873.68 thousand \textit{reais} in 2014~\cite{jeb:2014}. Consider that, in 2014, the exchange rate of the \textit{real} varied from $U\$2.19$ to $U\$2.72$ U.S. dollars. It is out of the scope of this paper to point \textit{fundo partid\'{a}rio} as the main culprit for Brazil's high level of party fragmentation, but it poses as a clear incentive for the creation of many parties in Brazil.

In spite of the fact that this work considers Brazil as its use case, the methods and results shown here can be easily replicated to other countries that have highly fragmented party systems. Carsten Anckar studied party system fragmentation in 77 countries~\cite{Anckar2000} and reported high levels of fragmentation for many countries besides Brazil, such as  Bolivia, Bulgaria, Denmark, Ecuador, Finland, France, Guatemala, India, Israel, Italy, Netherlands and Thailand. Nevertheless, it is important to point out that this work was motivated and eased by the Open Data initiative of the Brazilian government, that provides public data related to politics and also to many other areas, such as demographics, government spending, budget and road accidents.

\section{Conclusions}

In this work, I proposed the method \method to assess and reduce fragmentation in multi-party political systems. From roll votes data of partisans and their respective party leaders, \method redistributes the partisans into new parties considering two conflicting objectives: to minimize the number of parties and to maximize party discipline. When applied to Brazilian historical roll call data, \method was able to generate \pedror{$23$} distinct configurations that, compared with the \statusquo, have (i) a significantly smaller number of parties, (ii) higher discipline of partisans towards their parties and (iii) more even distributions of partisans into parties.  These results show that Brazil has and had many redundant parties, i.e., parties that are very similar ideologically. Thus, if today Brazil has one of the highest levels of party system fragmentation in the World~\cite{Limongi2009,Anckar2000}, this work proved it could be much lower. Finally, it is important to point out that \method is a general method and could be directly applied to analyze fragmentation in any of the many highly fragmented party systems that exists in the world~\cite{Anckar2000}.

%
%
%

\end{document}